\documentclass[11pt]{llncs}

\pdfoutput=1

\usepackage{amsmath, amssymb, enumitem,amsfonts,float,bm,stmaryrd,bm}
\usepackage{color,graphicx, subfig}
\usepackage[ruled,vlined,linesnumbered]{algorithm2e}
\usepackage{url}
\usepackage{verbatim,balance}
\usepackage{multirow,wrapfig}
\usepackage{rotating}
\usepackage{upgreek}

\usepackage[noadjust]{cite}

\usepackage{tikz}
\usetikzlibrary{arrows,backgrounds,decorations,decorations.pathmorphing,positioning,fit,automata,shapes,snakes,patterns}

\newcommand{\ie}{{i.e.}\xspace}

\newcommand{\etc}{{etc.}\xspace}

\newcommand{\secref}[1]{Sec.~\ref{sec:#1}}

\newcommand{\figref}[1]{Fig.~\ref{fig:#1}}

\newcommand{\tabref}[1]{Table~\ref{tab:#1}}

\newcommand{\ignore}[1]{}

\newtheorem{thm}{Theorem}[section]
\newcommand{\Nat}{\mathbb{N}}
\newcommand{\Reals}{\mathbb{R}}

\newcommand{\mtt}[1]{\mathtt{#1}}

\newcommand{\cP}{{\cal P}}
\newcommand{\cB}{{\cal B}}
\newcommand{\wcP}{\widehat{\cP}}
\newcommand{\wcB}{\widehat{\cB}}
\newcommand{\A}{{\cal A}}
\newcommand{\Pred}{Pred}
\newcommand{\pred}{\phi}

\newcommand{\Var}{V}
\newcommand{\var}{b}

\newcommand{\gVar}{GV}
\newcommand{\fVar}{FV}
\newcommand{\lVar}{LV}

\newcommand{\varname}{v}
\newcommand{\vname}{\varname}

\newcommand{\conc}{\gamma}
\newcommand{\Conc}{\Gamma}
\newcommand{\f}{F}
\newcommand{\cL}{{\cal L}}
\newcommand{\wcL}{\widehat{\cL}}
\newcommand{\cG}{{\cal G}}
\newcommand{\Stmt}{S}

\newcommand{\C}{{\cal C}}

\newcommand{\Cost}{{Cost}}
\newcommand{\cut}{{\ell}}
\newcommand{\cs}{{\cut}}
\newcommand{\ce}{{\cut'}}
\newcommand{\ch}{\grave{\lambda}}
\newcommand{\ci}{{\lambda}}
\newcommand{\cj}{{\lambda'}}
\newcommand{\stmt}{stmt}
\newcommand{\wstmt}{\widehat{\stmt}}
\newcommand{\Stmttype}{\Sigma}
\newcommand{\stmttype}{\sigma}
\newcommand{\wstmttype}{\widehat{\stmttype}}
\newcommand{\type}{\tau}
\newcommand{\ofp}{(\cP)}

\newcommand{\inscope}{inscope}
\newcommand{\entry}{entry}
\newcommand{\exit}{exit}
\newcommand{\N}{N}
\newcommand{\R}{E}
\newcommand{\VC}{VC}
\newcommand{\CRC}{CRC}

\newcommand{\cT}{{\cal T}}
\newcommand{\id}{id}

\newcommand{\funcDefinedBy}[1]{\ensuremath{\llbracket #1\rrbracket}}
\newcommand{\pc}{\rho}
\newcommand{\E}{{\xi}}

\newcommand{\identity}{id}
\newcommand{\cI}{{\cal I}}
\newcommand{\I}{{\cal I}}

\newcommand{\U}{{\cal U}}
\newcommand{\cU}{{\cal U}}
\newcommand{\update}{\mathbb{R}_{\U,\cL}}
\newcommand{\upd}{\mathbb{R}}
\newcommand{\temp}{\mathbb{E}_{\cT,\cL}}
\newcommand{\tel}{\mathbb{E}(\ell)}
\newcommand{\cE}{{\cal E}}
\newcommand{\cost}{c_{\U,\cL}}

\newcommand{\target}{{\delta}}
\newcommand{\st}{\varsigma}
\newcommand{\stack}{{\zeta}}

\newcommand{\setof}[2]{\{{#1}_1,\ldots,{#1}_{#2}\}}
\newcommand{\procind}{t}
\newcommand{\predind}{r}

\newcommand{\formalind}{k}
\newcommand{\assind}{m}
\newcommand{\upind}{d}

\newcommand{\la}{{\langle}}
\newcommand{\ra}{{\rangle}}
\newcommand{\en}[1]{\la{#1}\ra}

\newcommand{\assign}{\mathtt{:=}}
\newcommand{\true}{{\tt true}}
\newcommand{\false}{{\tt false}}

\newcommand{\tran}[1]{\xrightarrow{#1}}

\renewcommand{\labelitemi}{$\bullet$}

\newcommand{\err}{err}
\newcommand{\vs}{\varsigma}

\begin{document}
\title{Cost-Aware Automatic Program Repair}
\author{Roopsha Samanta\inst{1}, Oswaldo Olivo\inst{2}, \and E. Allen Emerson\inst{2}}
\institute{The University of Texas at Austin and IST Austria\\
\email{roopsha@ist.ac.at} \and
The University of Texas at Austin \\
\email{\{olivo,emerson\}@cs.utexas.edu}}

%\begin{document}
\maketitle

\begin{abstract}

We present a formal framework for repairing infinite-state, imperative,
sequential programs, with (possibly recursive) procedures and multiple
assertions; the framework can generate repaired programs by modifying the
original erroneous program in multiple program locations, and can ensure the
readability of the repaired program using user-defined expression templates;
the framework also generates a set of inductive assertions that serve as a
proof of correctness of the repaired program.  As a step toward integrating
programmer intent and intuition in automated program repair, we present a {\em
cost-aware} formulation --- given a cost function associated with permissible
statement modifications, the goal is to ensure that the total program
modification cost does not exceed a given repair budget. As part of our
predicate abstraction-based solution framework, we present a sound and complete
algorithm for repair of Boolean programs.  We have developed a prototype tool
based on SMT solving and used it successfully to repair diverse errors in
benchmark C programs. 

\end{abstract}

\section{Introduction}\label{sec:intro}
Program debugging --- the process of fault localization and error elimination
--- is an integral part of ensuring correctness in existing or evolving
software. Being essentially manual, program debugging is often a lengthy,
expensive part of a program's development cycle. There is an evident need for
improved formalization and mechanization of this process. However, program
debugging is hard to formalize --- there are multiple types of programming
mistakes with diverse manifestations, and multiple ways of eliminating a detected
error. Moreover, it is particularly challenging to assimilate and mechanize the
expert human intuition involved in the choices made in manual program
debugging.

In this paper, we present a {\em cost-aware} formulation of the automated
program debugging problem that addresses the above concerns. Our formulation
obviates the need for a separate fault localization phase by directly focusing
on error elimination, \ie, program repair. We fix a set $\U$ of {\em update
schemas} that may be applied to program statements for modifying them.  An
update schema is a compact description of a class of updates that may be
applied to a program statement in order to repair it. For instance, the update
schema $\mtt{assign} \mapsto \mtt{assign}$ permits replacement of the
assignment statement $x \, \assign \, y$ with other assignment statements such
as $x \, \assign \, x + y$ or $y \, \assign \, x + 1$, $\mtt{assign} \mapsto
\mtt{skip}$ permits deletion of an assignment statement, etc. In this paper,
$\U$ includes deletion of statements, replacement of assignment statements with
other assignment statements, and replacement of the guards of conditional and
loop statements with other guards. We assume we are given a {\em cost function}
that assigns some user-defined cost to each application of an update schema to
a program statement. Given an erroneous program $\cP$, a cost function $c$ and
a repair budget $\target$, the goal of {\em cost-aware automatic program
repair} is to compute a program $\wcP$ such that: $\wcP$ is correct, $\wcP$ is
obtained by modifying $\cP$ using a set of update schemas from $\U$ and the
total modification cost does not exceed $\target$. We postulate that this {\em
quantitative} formulation \cite{BCHJ09} is a flexible and convenient way of
incorporating user intent and intuition in automatic program debugging.  For
instance, the user can define appropriate cost functions to search for $\wcP$
that differs from $\cP$ in at most $\target$ statements, or to penalize any
modification within some {\em trusted} program fragment, or to favor the
application of a particular update schema over another, and so on. 

Our approach to cost-aware repair of imperative, sequential programs is based
on predicate abstraction \cite{GS97}, which is routinely used by verification
tools such as SLAM \cite{BR01}, SLAM2 \cite{BBKL10}, SATABS \cite{CKSY05}, etc.
for analyzing infinite-state programs.  These tools generate Boolean programs
which are equivalent in expressive power to pushdown systems and enjoy
desirable computational properties such as decidability of reachability
\cite{BR00}.  Inevitably, Boolean programs have also been explored for use in
automatic repair of sequential programs for partial correctness \cite{GBC06}
and total correctness \cite{SDE08}.  These papers, however, do not accommodate
a quantitative formulation of the repair problem and can only compute repaired
programs that differ from the original erroneous program in exactly one
expression. Moreover, these papers do not attempt to improve the {\em
readability} of the concrete program $\wcP$, obtained by concretizing a
repaired Boolean program. 

Our predicate abstraction-based approach to automatic program repair relaxes
the above limitations. Besides erroneous $\cP$, $c$, and $\target$, our
framework requires a Boolean program $\cB$, obtained from $\cP$ through
iterative predicate abstraction-refinement, such that $\cB$ exhibits a
non-spurious path to an error.  We present an algorithm which casts the
question of {\em repairability of $\cB$}, given $U$, $c$, and $\target$, as an
SMT query;  if the query is satisfiable, the algorithm extracts a correct
Boolean program $\wcB$ from the witness to its satisfiability. Along with
$\wcB$, we also extract a set of inductive assertions from the witness, that
constitute a proof of correctness of $\wcB$. This algorithm for Boolean program
repair is sound and complete, relative to $\U$, $c$, and $\target$.  A repaired
Boolean program $\wcB$, along with its proof, is concretized to obtain a
repaired concrete program $\wcP$, along with a proof of correctness.  However,
the concretized repairs may not be succinct or readable.  Hence, our framework
can also accept user-supplied templates specifying the desired syntax of the
modified expressions in $\wcP$ to constrain the concretization. 
 
Alternate approaches to automatic repair and synthesis of sequential programs
\cite{SRBE05,STBSS06,SGF10,KB11} that do not rely on abstract interpretations
of concrete programs, also often encode the repair/synthesis problem as a
constraint-solving problem whose solution can be extracted using SAT or SMT
solvers.  Except for \cite{SGF10}, these approaches, due to their bounded
semantics, are imprecise and cannot handle total correctness\footnote{Our
framework can be extended to handle total correctness by synthesizing ranking
functions along with inductive assertions.}. The authors in \cite{KB11} use SMT
reasoning to search for repairs satisfying user-defined templates; the
templates are needed not only for ensuring readability of the generated
repairs, but also for ensuring tractability of their inherently undecidable
repair generation query.  They also include a notion of minimal diagnoses,
which is subsumed by our more general cost-aware formulation.  Given
user-defined constraints specifying the space of desired programs and
associated proof objects, the scaffold-based program synthesis approach of
\cite{SGF10} attempts to synthesizes a program, along with a proof of total
correctness consisting of program invariants and ranking functions for loops.
In contrast to \cite{SGF10}, our framework only interacts with a user for
improving the readability of the generated repairs and for the cost function;
all predicates involved in the generation of the repaired Boolean program and
its proof are discovered automatically.  Besides the above, there have been
proposals for program repair based on computing repairs as winning strategies
in games \cite{JGB05}, abstraction interpretation \cite{LB12}, mutations
\cite{DW10}, genetic algorithms \cite{Arcuri08,LDFW12}, using contracts
\cite{WPFSBMZ10}, and focusing on data structure manipulations
\cite{SL11,ZGKM12}. There are also customized program repair engines for
grading and feedback generation for programming assignments, cf. \cite{SGS13}.
Finally, a multitude of algorithms \cite{ZH02,BNR03,JM11,CTBB11} have been
proposed for fault localization, based on analyzing error traces.  Some of
these techniques can be used as a preprocessing step to improve the efficiency
of our algorithm, at the cost of giving up on the completeness of the Boolean
program repair module. \\  

\noindent {\em Summary of contributions}: We define a new cost-aware
formulation of automatic program repair that can incorporate programmer
intuition and intent (\secref{prob}). We present a formal solution framework
(\secref{algo1} and \secref{algo2}) that can repair infinite-state, imperative,
sequential programs with (possibly recursive) procedures and multiple
assertions. Our method can modify the original erroneous program in multiple
program locations and can ensure the readability of the repaired program using
user-defined expression templates. If our method succeeds in generating a
repaired program $\wcP$, it generates a proof of $\wcP$'s correctness,
consisting of inductive assertions, that guarantees satisfaction of {\em all}
the assertions in the original program $\cP$. As part of our predicate
abstraction-based solution, we present a sound and complete algorithm for
repair of Boolean programs. Finally, we present experimental results for
repairing diverse errors in benchmark C programs using a prototype
implementation (\secref{expt}). 

%The paper is organized as follows. We begin with some preliminary definitions
%in \secref{prelims} and present our cost-aware formulation of automatic program
%repair in \secref{prob}. We explain the two main steps of our framework in
%\secref{algo1} and \secref{algo2}.  Finally, we present our experimental
%results in \secref{expt} and conclude with a discussion of possible extensions
%in \secref{exten}.

%While there has been extensive work on automated fault localization \cite{Blah}, most approaches are based %on customized heuristics, inspired by manual debugging experiences. A uniform, formal modeling of the %fault localization problem seems unlikely..

%\footnote{Note
%that most program synthesis frameworks can be adapted for use for
%program repair.} 

%\section{Motivating Example}\label{sec:example}
%\input{example.tex}

\section{Background}\label{sec:prelims}
\noindent{\bf Review: Predicate Abstraction}.  Predicate abstraction
\cite{GS97,BR00} is an effective approach for model checking infinite-state
imperative programs with respect to safety properties. This technique computes
a finite-state, {\em conservative} abstraction of a concrete program $\cP$ by
partitioning $\cP$'s state space based on the valuation of a finite set $\Pred
= \setof{\pred}{\predind}$ of predicates. The resulting abstract program is
termed a {\em Boolean program} $\cB$ (see \figref{runex.P} and \figref{runex.B}): 
the control-flow of $\cB$ is the same as
that of $\cP$ and the set $\Var = \setof{\var}{\predind}$ of variables of $\cB$
are Boolean variables, where for each $i \in [1,\predind]$, the Boolean
variable $\var_i$ represents the predicate $\pred_i$.  Given a concrete program
$\cP$, the overall {\em counterexample-guided abstraction refinement} method
proceeds as follows.  In step one, an initial Boolean program $\cB$ is computed
and in step two, $\cB$ is model-checked with respect to its specification. If
$\cB$ is found to be correct, the  method concludes that $\cP$ is correct.
Otherwise, an abstract counterexample path leading to some violated assertion
in $\cB$ is computed and examined for feasibility in $\cP$. If found feasible,
the method terminates, reporting an error in $\cP$.  If found infeasible, in
step three, $\cB$ is refined into a new Boolean program $\cB'$ that eliminates
the spurious counterexample. Thereafter, steps two and three are repeated, as
needed. Note that the overall method is incomplete - it may not always be able
to possible to compute a suitable refinement that eliminates a spurious
counterexample or to check if an abstract counterexample is indeed spurious. 

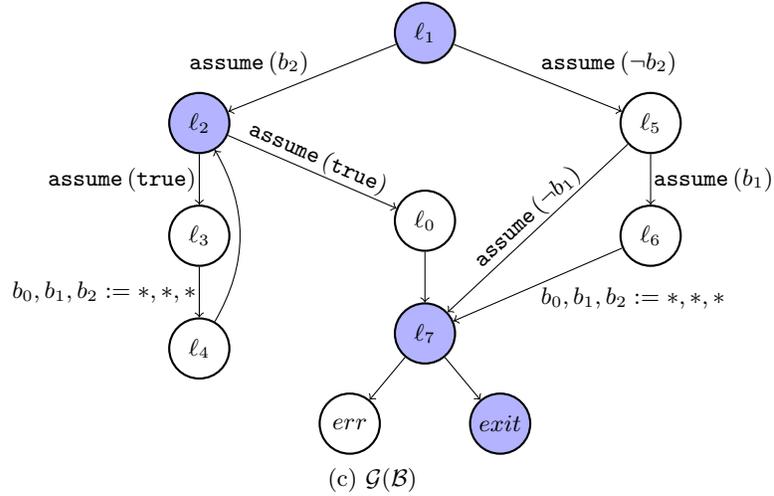
\begin{figure}[ht!b]
\centering
%\footnotesize{
\subfloat[$\cP$]{
\label{fig:runex.P}
\begin{minipage}{0.33\linewidth}
\footnotesize{
{\tt
$\mtt{main}()$ \{\\
\hspace*{3mm} $\mtt{int} \; x$;\\
\hspace*{3mm} $\ell_1:$ $\mtt{if}$ $(x\leq0)$\\
\hspace*{3mm} $\ell_2:$ \hspace*{3mm} $\mtt{while}$ $(x < 0)$\{\\
\hspace*{3mm} $\ell_3:$ \hspace*{6mm} $x \; \assign \; x + 2$;\\
\hspace*{3mm} $\ell_4:$ \hspace*{6mm} $\mtt{skip}$;\\
\hspace*{15mm} $\}$\\
\hspace*{10mm} $\mtt{else}$\\
\hspace*{3mm} $\ell_5:$ \hspace*{3mm} $\mtt{if}$ $(x==1)$\\
\hspace*{3mm} $\ell_6:$ \hspace*{6mm} $x \; \assign \; x - 1$;\\
\hspace*{3mm} $\ell_7:$ $\mtt{assert}$ $(x>1)$;\\
$\}$
}}
\end{minipage}}
%\hspace{4mm}
\subfloat[$\cB$]{
\label{fig:runex.B}
\begin{minipage}{0.68\linewidth}
\footnotesize{
{\tt
$\mtt{main}()$ \{\\
\hspace*{3mm} $/* \conc(b_0)\, = x \leq 1, \; \conc(b_1)\, =\, x==
                  1,\; \conc(b_2) \,=\, x \leq 0 */$\\
\hspace*{3mm} $\mtt{Bool} \; b_0,b_1,b_2 \; \assign \; *,*,*$;\\
\hspace*{3mm} $\ell_1:$ $\mtt{if}$ $(\neg b_2)$ $\mtt{then}$ $\mtt{goto} \;
\ell_5$; \\
\hspace*{3mm} $\ell_2:$ $\mtt{if}$ $(*)$ $\mtt{then}$
$\mtt{goto} \; \ell_0$; \\
\hspace*{3mm} $\ell_3:$ $b_0,b_1,b_2 \; \assign \; *,*,*$;\\
\hspace*{3mm} $\ell_4:$ $\mtt{goto} \; \ell_1$;\\
\hspace*{3mm} $\ell_0:$ $\mtt{goto} \; \ell_7$;\\
\hspace*{3mm} $\ell_5:$ $\mtt{if}$ $(\neg b_1)$ $\mtt{then}$ $\mtt{goto}
\; \ell_7$;\\
\hspace*{3mm} $\ell_6:$ $b_0,b_1,b_2 \; \assign \; *,*,*$;\\
\hspace*{3mm} $\ell_7:$ $\mtt{assert}$ $(\neg b_0)$;\\
$\}$
}}
\end{minipage}}

\subfloat[$\cG(\cB)$]{
\label{fig:runex.BTG}
\begin{minipage}{0.8\linewidth}
\begin{tikzpicture}[->]
\footnotesize{
    \tikzstyle{state}=[draw,circle,minimum size=8mm,thick,inner sep=1pt,text=black]
    \tikzstyle{lbl}=[font=\fontsize{9}{9}\selectfont,inner
        sep=1pt,minimum height=2mm]
    
    \node[state,fill=blue!30]              (s1) {$\ell_1$};
    \node[state,fill=blue!30,node distance=12mm, below of=s1,xshift=-30mm] (s2)
    {$\ell_2$};
    \node[state, node distance=15mm, below of=s2] (s3) 
    {$\ell_3$};
    \node[state, node distance=15mm, below of=s3] (s4) {$\ell_4$};

    \node[state, node distance=25mm, below of=s1] (s0) {$\ell_0$};

    \node[state, node distance=12mm, below of=s1,xshift=30mm] (s5)
    {$\ell_5$};
    \node[state, node distance=15mm, below of=s5,xshift=0mm] (s6)
    {$\ell_6$};

    \node[state,fill=blue!30,node distance=40mm, below of=s1] (s7) {$\ell_7$};
    \node[state, node distance=12mm, below of=s7,xshift=-10mm] (s8)
    {$err$};
    \node[state,fill=blue!30,node distance=12mm, below of=s7,xshift=10mm] (s9)
    {$exit$};

    \path[->] (s1) edge node [lbl,above left] {$\mtt{assume} \, (b_2)$} (s2);
    \path[->] (s2) edge node [lbl,left] {$\mtt{assume} \, (\true)$} (s3);
    \path[->] (s3) edge node [lbl,left] {$b_0,b_1,b_2 \; \assign \;
    *,*,*$} (s4);
    \path[->] (s4) edge[bend right] (s2);

    \path[->] (s1) edge node [lbl,above right]{$\mtt{assume} \, (\neg b_2)$} (s5);
    \path[->] (s5) edge node [lbl,right]{$\mtt{assume} \, (b_1)$} (s6);

    \path[->](s2) edge node [lbl,above,sloped] {$\mtt{assume} \,
    (\true)$} (s0);
    \path[->] (s0) edge (s7);
    \path[->] (s5) edge node [lbl,above,sloped,midway] {$\mtt{assume}
    \,(\neg b_1)$} (s7);
    \path[->] (s6) edge node[lbl,below right] {$b_0,b_1,b_2 \; \assign \; *,*,*$} (s7);

    \path[->] (s7) edge (s8);
    \path[->] (s7) edge (s9);
}
\end{tikzpicture}
\end{minipage}}
%\end{center}
\caption{An example concrete program $\cP$, a corresponding Boolean
program $\cB$ and $\cB$'s transition graph}
\label{fig:runexPBTG}
\end{figure}

In our work, the interesting case is when the method terminates reporting an
error. Henceforth, we fix a concrete program $\cP$, and a corresponding Boolean
program $\cB$ that exhibits a non-spurious counterexample path.  Let 
$\setof{\pred}{\predind}$ denote the set of predicates used in the abstraction
of $\cP$ into $\cB$, where each predicate is a quantifier-free first order
expression over the variables of $\cP$.  Let $\setof{\var}{\predind}$
denote the corresponding Boolean variables of $\cB$.  Let $\conc$ denote the
mapping of Boolean variables to their respective predicates: for each $i \in
[1,\predind]$, $\conc(\var_i) = \pred_i$.  The mapping $\conc$ can be extended
in a standard way to expressions over the Boolean variables in $\Var$.\\

\noindent {\bf Program Syntax}.  For our technical presentation, 
we fix a common, simplified syntax for sequential concrete and
abstract programs. A partial definition of this syntax is shown in
\figref{syntax}. In the syntax, $\varname$ denotes a variable, $\en{type}$
denotes the type of a variable, $\f$ denotes a procedure, $\ell$ denotes a
statement label or location, $\en{expr}$ denotes a well-typed expression, and
$\en{bexpr}$ denotes a Boolean-valued expression.

\begin{figure}
\begin{center}
\begin{tabular}[t]{|llp{8.7cm}|}
\hline
$\en{pgm}$ & $::=$ & $\en{vardecl} \, \en{proclist}$\\
$\en{vardecl}$ & $::=$ & $\mtt{decl} \; \varname :  \en{type} \bm{;}$ $|$
$\en{vardecl} \, \en{vardecl}$\\

$\en{proclist}$ & $::=$ & $\en{proc} \, \en{proclist}$ $|$ $\en{proc}$\\

$\en{proc}$ & $::=$ & $\f(\varname_1,\ldots,\varname_\formalind) \; \mtt{begin} \, \en{vardecl}
\, \en{stmtseq} \, \mtt{end}$\\

$\en{stmtseq}$ & $::=$ & $\en{labstmt} \, \bm{;} \, \en{stmtseq}$\\

$\en{labstmt}$ & $::=$ & $\en{stmt}$ $|$ $\ell: \en{stmt}$\\

$\en{stmt}$ & $::=$ & $\mtt{skip}$  $|$ $\varname_1,\ldots, \varname_m \, \bm{\assign} \, \en{expr_1},\ldots,\en{expr_m}$
\newline
$|$ $\mtt{if}\, (\en{bexpr}) \, \mtt{then} \, \en{stmtseq} \,
\mtt{else} \, \en{stmtseq} \, \mtt{fi}$ 
\newline 
$|$ $\mtt{while} \, (\en{expr}) \, \mtt{do} \, \en{stmt} \, \mtt{od}$
$|$ $\mtt{assume} \, (\en{bexpr})$
\newline
$|$ $\mtt{call} \, \f(\en{expr_1},\ldots,\en{expr_\formalind})$
$|$ $\mtt{return}$
\newline
$|$ $\mtt{goto} \; \ell_1\, \mtt{or} \ldots \mtt{or}\, \ell_n$ $|$ 
$\mtt{assert} \,  (\en{bexpr})$\\

\hline
\end{tabular}
\end{center}
\caption{Programming language syntax}
\label{fig:syntax}
\end{figure}

Thus, a concrete or an abstract (Boolean) program consists of a 
declaration of global variables, followed by a list of procedure definitions; a
procedure definition consists of a declarations of local variables,
followed by a sequence of labeled statements; a statement is a $\mtt{skip}$,
(parallel) assignment, conditional, loop, $\mtt{assume}$,
(call-by-value) procedure $\mtt{call}$, $\mtt{return}$, $\mtt{goto}$
or $\mtt{assert}$ statement.

We make the following assumptions: (a) there is a distinguished initial
procedure $\mathtt{main}$, which is not called from any other procedure, (b)
all variable and formal parameter names are globally unique, (c) the number of
actual parameters in a procedure call matches the number of formal parameters
in the procedure definition, (d) $\mtt{goto}$ statements are not used
arbitrarily; they are used only to simulate the flow of control in structured
programs, (e) the last statement in the loop body of every $\mtt{while}$
statement is a $\mtt{skip}$ statement,  and (f) $\en{type}$ includes integers
and Booleans.  In addition, for Boolean programs, we assume: (a) all variables
and formal parameters are of $\en{type}$ Boolean and (b) all expressions -
$\en{expr}$, $\en{bexpr}$ - are Boolean expressions defined as follows:\\

\begin{tabular}[h]{l l p{9cm}}
$\en{bexpr}$ & $::=$ & $*$ $|$ $\en{detbexpr}$\\
$\en{detbexpr}$ & $::=$ & $\true$ $|$ $\false$ $|$ $\var$
\newline
$|$ $\neg\en{detbexpr}$ $|$ $\en{detbexpr} \Rightarrow \en{detbexpr}$
\newline
$|$ $\en{detbexpr} \vee \en{detbexpr}$ $|$
$\en{detbexpr} \wedge \en{detbexpr}$
\newline
$|$ $\en{detbexpr} = \en{detbexpr}$
$|$ $\en{detbexpr} \neq \en{detbexpr}$,
\end{tabular}

\noindent where $b$ is a Boolean variable.  Thus, a Boolean expression is
either a deterministic Boolean expression or the expression $*$, which
nondeterministically evaluates to $\true$ or $\false$\footnote{In practice, a
nondeterministic Boolean expression is any Boolean expression containing $*$ or
the expression $\mtt{choose}(e_1,e_2)$, with $e_1$, $e_2$ being deterministic
Boolean expressions (if $e_1$ is $\true$, $\mtt{choose}(e_1,e_2)$ evaluates to
$\true$, else if $e_2$ is $\true$, $\mtt{choose}(e_1,e_2)$ evaluates to
$\false$, else $\mtt{choose}(e_1,e_2)$ evaluates to $*$). While we handle
arbitrary nondeterministic Boolean expressions in our prototype tool (see
\secref{expt}), we only consider $*$ expressions in our exposition for
simplicity.}. We assume that $*$ expresses a {\em fair} nondeterministic
choice, \ie, $*$ does not permanently evaluate to the same value. We assume
that Boolean expressions in $\mtt{assume} \, (\en{bexpr})$ and $\mtt{assert} \,
(\en{bexpr})$ statements are always deterministic. Thus, a concrete program
contains no nondeterministic expressions, and a Boolean program contains
nondeterministic expressions only in the RHS of assignment statements.

%Note that conditional and loop statements can be
%expressed using $\mtt{assume}$ and $\mtt{goto}$ statements in a
%straight-forward manner. For instance, the conditional statement
%$\mtt{if}$ $(*)$ $s_{\mtt{if}}$ $\mtt{else}$ $s_{\mtt{else}}$
%$\mtt{fi}$ can be encoded as:
%\begin{center}
%\begin{tabular}[h]{l l l l}
%$0: \, \mtt{goto}  \; 1 \, \mtt{or} \, 4;$  &\hspace{1in}$4: \,  \mtt{assume} \, (\true);$ \\
%$1: \, \mtt{assume} \, (\true);$            &\hspace{1in}$5: \,s_{\mtt{else}};$ \\
%$2: \, s_{\mtt{if}};$                       &\hspace{1in}$6: \, \mtt{goto} \; 7;$\\
%$3: \, \mtt{goto} \; 7;$                    &\hspace{1in}$7: $ \\
%\end{tabular}\\
%\end{center}
%\noindent and the loop statement $\mtt{while}$ $(g)$ $\mtt{do}$ $s_{\mtt{loop}}$
%$\mtt{od}$ can be encoded as:
%\begin{center}
%\begin{tabular}[h]{l l}
%$0: \, \mtt{goto}  \; 1 \, \mtt{or} \, 4;$  &\hspace{1in}$4: \,  \mtt{assume} \, (\neg g);$\\
%$1: \, \mtt{assume} \, (g);$                &\hspace{1in}$5: \,  \mtt{goto} \; 6;$ \\
%$2: \, s_{\mtt{loop}};$                     &\hspace{1in}$6: $ \\
%$3: \, \mtt{goto} \; 0;$ & \\
%\end{tabular}\\
%\end{center}
%\noindent Observe that in the above encodings, the Boolean expressions in the
%$\mtt{assume}$ statements are always deterministic. Thus, a concrete
%program contains no nondeterministic expressions, and a Boolean
%program contains nondeterministic expressions only in the RHS of
%assignment statements.

Note that the above syntax does not permit return values from
procedures. However, return values can be easily modeled using extra
global variables. Hence, this syntax simplification does not affect
the expressivity of the programming language. Indeed, the above syntax
is quite general.\\

\noindent{\em Notation}. Let us fix some notation before we proceed.  For program $\cP$, let
$\{F_0,\ldots, F_\procind\}$ be its set of procedures with $F_0$
being the $\mtt{main}$ procedure, and let $\gVar(\cP)$ denote the set
of global variables.  For procedure $F_i$, let $\Stmt_i$ and $\cL_i$
denote the sets of statements and locations, respectively, and let
$\fVar_{i}$ and $\lVar_{i}$ denote the sets of formal parameters and
local variables, respectively, with $\fVar_{i} \subseteq \lVar_{i}$.
Let $\Var\ofp = \gVar\ofp \, \cup \, \bigcup_{i=1}^{t} \lVar_{i}$
denote the set of variables of $\cP$, and $\cL\ofp = \bigcup_{i=1}^t
\cL_i$  denote the set of locations of $\cP$. For a location $\ell$
within a procedure $F_i$, let $\inscope(\ell) = \gVar\ofp \, \cup
\, \lVar_{i}$ denote the set of all variables in $\cP$ whose scope
includes $\l$.  We denote by $\stmt(\ell)$, $formal(\ell)$ and
$local(\ell)$ the statement at $\ell$ and the sets of formal
parameters and local variables of the procedure containing $\ell$,
respectively. We denote by $\entry_i \in \cL_i$ the location of the
first statement in $F_i$. When the context is clear, we simply use
$\Var$, $\cL$ instead of $\Var\ofp$, $\cL\ofp$ etc.\\

\noindent{\bf Transition Graphs}.  In addition to a {\em textual}
representation, we will often find it convenient to use a {\em transition
graph} representation of programs.  The transition graph representation of
$\cP$, denoted $\cG\ofp$, comprises a set of labeled, rooted, directed graphs
$\cG_0,\ldots ,\cG_{\procind}$, which have exactly one node, $\err$, in common.
Informally, the $i^{th}$ graph $\cG_i$ captures the flow of control in
procedure $F_i$ with its nodes and edges labeled by locations and corresponding
statements of $F_i$, respectively. To be more precise, $\cG_i = (\N_i, Lab_i,
\R_i)$, where the set of nodes $\N_i$, given by $\cL_i \, \cup \,  \exit_i \,
\cup \err$, includes a unique entry node $\entry_i$, a unique exit node
$\exit_i$ and the error node $\err$, the set of labeled edges $\R_i \,
\subseteq \, \N_i \times Lab_i \times \N_i$ is defined as follows: for all
$\ell,\ell' \in \N_i$, $(\ell, \vs, \ell') \in \, \R_i$ iff:

{\renewcommand\labelitemi{-}
\begin{itemize}
\item $\stmt(\ell)$ is an assignment, $\mtt{assume} \,
(g)$ or $\mtt{call} \, \f(e_1,\ldots,e_\formalind)$ statement,
$\ell'$ is the next sequential location\footnote{The next
sequential location of the last statement in the $\mtt{then}$ or
$\mtt{else}$ branch of a conditional statement is the location
following the conditional statement. The next sequential
location of the last statement in the $\mtt{main}$ procedure is
$\exit_0$.} in $F_i$ after $\ell$ and $\vs = \stmt(\ell)$, or,

\item $\stmt(\ell)$ is a $\mtt{skip}$ statement and either (a)
$\stmt(\ell)$ is the last statement in the loop body of
a statement $\ell': \, \mtt{while} \, (g)$ and $\vs$ is the empty
label, or, (b) $\ell'$ is the next sequential location in $F_i$
after $\ell$ and $\vs$ is the empty label, or,

\item $\stmt(\ell)$ is $\mtt{if} \, (g)$, and either (a) $\ell'$,
denoted $Tsucc(\ell)$, is the location of the first statement in the
$\mtt{then}$ branch and $\vs = \mtt{assume} \, (g)$, or, (b) $\ell'$,
denoted $Fsucc(\ell)$, is the location of the first statement in the
$\mtt{else}$ branch and $\vs = \mtt{assume} \, (\neg g)$, or,

\item $\stmt(\ell)$ is $\mtt{while} \, (g)$, and either (a) $\ell'$,
denoted $Tsucc(\ell)$, is the location of the first statement in the
$\mtt{while}$ loop body and $\vs = \mtt{assume} \, (g)$, or, (b)
$\ell'$, denoted $Fsucc(\ell)$, is the next sequential location in
$F_i$ after the end of the $\mtt{while}$ loop body and $\vs =
\mtt{assume} \, (\neg g)$, or,

\item $\stmt(\ell)$ is $\mtt{assert} \, (g)$, and either $\ell'$,
denoted $Tsucc(\ell)$, is the next sequential location in $F_i$ after
$\ell$ and $\vs$ is the empty label, or, (b) $\ell'$, denoted
$Fsucc(\ell)$, is the node $\err$ and $\vs$ is the empty label, or,

\item $\stmt(\ell)$ is a $\mtt{goto}$ statement that includes the label
$\ell'$, and $\vs$ is the empty label, or,

\item $\stmt(\ell)$ is a $\mtt{return}$ statement, $\ell' = exit_i$
and $\vs = \mtt{return}$.
\end{itemize}

Let $succ(\ell)$ denote the set $\{\ell': (\ell, \vs, \ell')
\in \R_i$\} for some $i \in [0,\procind]$. A {\em path} $\pi$ in $\cG_i$ is a sequence
of labeled connected edges; with some overloading of notation, we
denote the sequence of statements labeling the edges in $\pi$ as
$\stmt(\pi)$. Not that every node in $\cG_i$ is on some path between $\entry_i$
and $\exit_i$.

The transition graph of Boolean programs can be defined similarly (see
\figref{runex.BTG}).  The main modification is as follows.  In defining the
set of labeled edges $\R_i$ of graph $\cG_i = (\N_i,Lab_i,\R_i)$ in the
transition graph representation $\cG(\cB)$ of $\cB$, for $\ell \in \N_i$ with
$\stmt(\ell)$ given by $\mtt{if} \, (*)$ or $\mtt{while} \, (*)$,
$Tsucc(\ell)$, $Fsucc(\ell)$ are defined as above, but the
labels $\vs_1$, $\vs_2$ in $(\ell,\vs_1,Tsucc(\ell))$,
$(\ell,\vs_2,Fsucc(\ell))$ are each set to $\mtt{assume} \, (\true)$. \\

\noindent{\bf Program Semantics and Correctness}. 
\noindent Given a set $\Var_s \subseteq \Var$ of variables, a {\em valuation}
$\Omega$ of $\Var_s$ is a function that maps each variable in $\Var_s$
to an appropriate value of its {\em type}.  $\Omega$ can be naturally
extended to map well-typed expressions over variables to values.

An operational semantics can be defined for our programs by formalizing the
effect of each type of program statement on a program {\em configuration}.  A
configuration $\eta$ of a program $\cP$ is a tuple of the form $(\ell, \Omega,
\stack)$, where  where $\ell \in
\bigcup_{i=0}^\procind \N_i$, $\Omega$ is a valuation of the variables in
$\inscope(\ell)$\footnote{For $\ell = \exit_i$,
$inscope(\ell) = \gVar \cup \lVar_i$, and for $\ell = \err$,
$\inscope(\ell)$ is undefined.}, and $\stack$ is a stack of elements. Each element of $\stack$
is of the form $(\widetilde{\ell}, \widetilde{\Omega})$, where
$\widetilde{\ell} \in \cL_i$ for some $i$ and $\widetilde{\Omega}$ is a
valuation of the variables in $local(\widetilde{\ell})$. A program {\em state}
is a pair of the form $(\ell,\Omega)$, where $\ell$ and $\Omega$ are as defined
above; thus a program state excludes the stack contents.  A configuration $(\ell,
\Omega, \stack)$ of $\cP$ is called an initial configuration if $\ell =
\entry_0$ is the entry node of the $\mtt{main}$ procedure and $\stack$ is the
empty stack. We use $\eta \leadsto \eta'$ to denote that $\cP$ can transition
from configuration $\eta = (\ell, \Omega, \stack)$ to configuration $\eta' =
(\ell', \Omega', \stack')$; the transitions rules for each type of program
statement at $\ell$ and for exit nodes of procedures are presented in
\figref{trans}.

\begin{figure}[t]
\footnotesize
\begin{tabular}{|p{1.1cm}|p{2.8cm}||p{8.8cm}|}
\hline
\multicolumn{2}{|c||}{Cases}  & $(\ell, \Omega, \stack) \leadsto (\ell', \Omega', \stack')$ if:\\
\hline
\multirow{7}{*}{$\stmt(\ell)$:} &
 $\mtt{skip}$ & \multirow{2}{*}{$\ell' = succ(\ell)$, $\Omega' = \Omega$ and $\stack' = \stack$}\\
& $\mtt{return}$ & \\
\cline{2-3}
& $\mtt{goto} \, \ell_1 \, \mtt{or} \ldots \mtt{or} \, \ell_n$ & $\ell' \in succ(\ell)$, $\Omega' = \Omega$ and $\stack' = \stack$\\
\cline{2-3}
& $\mtt{assume} \, g$ & $\Omega(g) = \true$, $\ell' = succ(\ell)$,
$\Omega' = \Omega$ and $\stack' = \stack$\\
\cline{2-3}
& $\mtt{if} \, g$ & either $\Omega(g)= \true$, $\ell' = Tsucc(\ell)$,
$\Omega' = \Omega$ and $\stack' = \stack$, or,\\
& $\mtt{while} \, g$ & $\Omega(g)= \false$, $\ell' = Fsucc(\ell)$,
$\Omega' = \Omega$ and $\stack' = \stack$\\
\cline{2-3}
& $\mtt{assert} \, g$ & either $\Omega(g)= \true$, $\ell' = Tsucc(\ell)$, $\Omega' =
\Omega$ and $\stack' = \stack$, or, \\
& &  $\Omega(g) = \false$ and $\ell' = Fsucc(\ell) = \err$\\
\cline{2-3}
& $\vname_1,\ldots,\vname_m \, \assign$
\newline
$e_1,\ldots,e_m$ & 
$\ell' = succ(\ell)$,  
\newline 
$\forall i \in[1,m]: \Omega'(\varname_i) = \Omega(e_i)$, 
\newline
$\forall \varname \not\in \{\varname_1,\ldots, \varname_m\}:
\Omega'(\varname) = \Omega(\varname)$ and
\newline
$\stack' = \stack$\\
\cline{2-3}
& $\mtt{call} \, F_j(e_1,\ldots, e_k)$ & $\ell' = \entry_j$, 
\newline 
$\forall \varname_i \in formal(\ell'): \Omega'(\varname_i) = \Omega(e_i)$, 
\newline
$\forall \varname \in \gVar(\cP): \Omega'(\varname) =
\Omega(\varname)$ and
\newline 
$\stack' = (succ(\ell), \Delta) .\stack$, where 
$\forall \varname \in local(\ell): \Delta(\varname) = \Omega(\varname)$\\
\hline
$\ell$: & $\exit_j$ & $\ell' = \ell_{ret}$, 
\newline 
$\forall \varname \in local(\ell'): \Omega'(\varname) = \Delta(\varname)$, 
\newline
$\forall \varname \in \gVar(\cP): \Omega'(\varname) =
\Omega(\varname)$ and
\newline 
$\stack = (\ell_{ret}, \Delta).\stack'$\\
\hline
\end{tabular}
\caption{Transition rules for $(\ell, \Omega, \stack) \leadsto (\ell', \Omega', \stack')$.}
\label{fig:trans}
\end{figure}

Let us take a closer look at the last two transition rules in \figref{trans} -
the only transition rules that affect the stack contents.  Upon {\em execution}
of the statement $\mtt{call} \; F_j(e_1,\ldots, e_k)$ in program
configuration  $(\ell, \Omega, \stack)$, control moves to the entry node of the
called procedure $F_j$; the new valuation $\Omega'$ of program variables is
constrained to agree with $\Omega$ on the values of all global variables, and
maps the formal parameters of $F_j$ to the values of the actual arguments
according to $\Omega$; finally, the element $(succ(\ell), \Delta)$ is pushed
onto the stack, where $succ(\ell)$ is the location to which control returns
after $F_j$ completes execution and $\Delta$ is a valuation of all local
variables of the calling procedure, as recorded in $\Omega$.  The last
transition rule in \figref{trans} captures the return of control to the calling
procedure, say $F_i$, after completion of execution of a called procedure,
say $F_j$; the top of the stack element $(\ell_{ret},\Delta)$ is removed and
is used to retrieve the location $\ell_{ret}$ of $F_i$ to which control must
return as well the valuation $\Delta$ of the local variables of $F_i$; the
new valuation $\Omega'$ of program variables is constrained to agree with
$\Omega$ on the values of all global variables, and to agree with $\Delta$ on
the values of all local variables of $F_i$.

An {\em execution path} of program $\cP$ is a sequence of configurations, $\eta
\leadsto \eta' \leadsto \eta'' \leadsto \ldots$, obtained by repeated
application of the transition rules from \figref{trans}, starting from an
initial configuration $\eta$. Note that an execution path may be finite or
infinite. The last configuration $(\ell, \Omega,\stack)$ of a finite execution
path may either be a {\em terminating configuration} with $\ell = exit_0$, or
an {\em error configuration} with $\ell = \err$, or a {\em stuck configuration}
with $\ell \neq exit_0$. An execution path ends in a stuck configuration $\eta$
if none of the transition rules from \figref{trans} are applicable to $\eta$.
In particular, notice that notice that a transition from configuration
$(\ell, \Omega, \stack)$ with $\stmt(\ell)$ being $\mtt{assume} \,
(g)$ is defined only when $\Omega(g) = \true$.

The operational semantics of Boolean programs can be defined similarly.  The
main modifications are as follows.  For $\stmt(\ell)$ given by $\mtt{if} \,
(*)$ or $\mtt{while} \, (*)$, we say $(\ell, \Omega, \stack) \leadsto (\ell',
\Omega', \stack')$ if $\ell' \in succ(\ell)$, $\Omega' = \Omega$ and $\stack' =
\stack$. For  $\stmt(\ell)$ given by the assignment statement
$b_1,\ldots,b_j,\ldots,b_m \, \assign$ $e_1,\ldots,*,\ldots,e_m$, we say
$(\ell, \Omega, \stack) \leadsto (\ell', \Omega', \stack')$ if $\ell' =
succ(\ell)$, $\stack' = \stack$, $\forall i \in \{1,\ldots,j-1,j+1,\ldots,m]:
\Omega'(b_i) = \Omega(e_i)$, $\forall \varname \not\in \{b_1,\ldots, b_m\}:
\Omega'(\varname) = \Omega(\varname)$, and either $\Omega'(b_j) = \true$ or
$\Omega'(b_j) = \false$. This transition rule can be extended to handle other
scenarios such as assignment statements with multiple $*$ expressions in the
RHS, and $\mtt{call}$ statements with $*$ expressions in the actual arguments.

An assertion in program $\cP$, is a statement of the form $\ell: \mtt{assert}
\, (g)$, with $g$ being a  quantifier-free, first order expression representing
the expected values of the program variables in $inscope(\ell)$ at $\ell$.  We
will use the term assertion to denote both the statement $\ell: \mtt{assert} \,
(g)$ as well as the expression $g$.  We say a program configuration $(\ell,
\Omega,\stack)$ satisfies an assertion, if the embedded variable valuation
$\Omega$ satisfies the same.

Given a program $\cP$ annotated with a set of assertions, $\cP$ is {\em
partially correct} iff every finite execution path of $\cP$ ends in a
terminating configuration. We say $\cP$ is {\em totally correct} iff every
execution path is finite and ends in a terminating configuration. In what
follows, we assume that all programs are annotated with a set of assertions.

In specifying correctness for Boolean programs, we interpret the nondeterminism
in them as Dijkstra's {\em demonic} nondeterminism \cite{D76}.  Given a program
$\cB$ annotated with a set of assertions, $\cB$ is partially correct iff every
{\em finite} execution path of $\cB$ ends in a terminating configuration for
all nondeterministic choices that $\cB$ might make. $\cB$ is totally correct
iff every execution path is finite and ends in a terminating configuration, for
all nondeterministic choices that $\cB$ might make.

Unless otherwise specified, an {\em incorrect} program is one that is
not partially correct.

\noindent{\em Remark}: While we found it convenient to define Boolean programs
as we did above, it is worth noting that formalisms such as pushdown systems
\cite{BEM97} and recursive state machines \cite{ABEGRY05} are equivalent to
Boolean programs.

%\section{Program Syntax, Semantics and
%Correctness}\label{sec:programs}
%\input{programs.tex}

\section{Cost-aware Program Repair}\label{sec:prob}

\subsection{The Problem}

Let $\Stmttype$ denote the set of {\em statement types} in program $\cP$. 
As can be seen from \figref{runex.BTG}, it suffices to consider the set of 
statement types given by $\Stmttype = \{\mtt{skip},  \mtt{assign}, \mtt{assume},
\mtt{assert}, \mtt{call}, \mtt{return}, \mtt{goto}\}$. 
Given a statement $\st$,
let $\type(\st)$ be an element of $\Stmttype$ denoting the statement
type of $\st$.  Let $\U = \{u_0,u_1,\ldots,u_\upind\}$ be a set of
permissible, statement-level {\em update schemas}: $u_0 = id$ is the
{\em identity} update schema that maps every statement to itself, and
$u_i$, $i \in [1,m]$, is a function $\stmttype \mapsto \wstmttype$,
$\stmttype, \wstmttype \in \Stmttype \setminus \{\mtt{assert}\}$, that
maps a statement type to a statement type. For each update schema $u$,
given by $\stmttype \mapsto \wstmttype$, we say $u$ can be {\em
applied} to statement $\st$ to get statement ${\widehat \st}$ if
$\type(\st) = \stmttype$; $\type({\widehat \st})$ is then given by $\wstmttype$. 
For example, $u$, given by $\mtt{assign}  \mapsto
\mtt{assign}$, can be applied to the assignment statement $\ell: x \,
\assign \, y$ to get other assignment statements such $\ell: x \,
\assign \, x + y$, $\ell: y \, \assign \, x + 1$ \etc Notice that
update schemas in $\U$ do not affect the label of a statement, and
that we do not permit any modification of an $\mtt{assert}$ statement.
In this paper, we fix the following set of permissible update
schemas for
programs:
\begin{align}\label{upschema}
\U = \{\identity, \, &\mtt{assign} \mapsto \mtt{assign}, \,
\mtt{assign} \mapsto \mtt{skip}, \, \mtt{assume} \mapsto
\mtt{assume},\\
&\mtt{call} \mapsto \mtt{call}, \, \mtt{call} \mapsto \mtt{skip}\}. \notag
\end{align}

We extend the notion of a statement-level update to a program-level
update as follows.  For programs $\cP$, $\wcP$, let the respective
sets  of locations be $\cL$, $\wcL$ and let $\stmt(\ell)$,
$\wstmt(\ell)$ denote the respective statements at location $\ell$.
Let $\update: \cL \mapsto \U$ be a function that maps each location of
$\cP$ to an update schema in $\U$. We say $\wcP$ is a $\update$-update
of $\cP$ iff $\cL = \wcL$ and for each $\ell \in \cL$, $\wstmt(\ell)$
is obtained by applying $\update(\ell)$ on $\stmt(\ell)$.

Let $\cost: \U \times \cL \to \Nat$ be a cost function that maps a tuple,
consisting of a statement-level update schema $u$ and a location $\ell$ of
$\cP$, to a certain cost. Thus, $\cost(u,\ell)$ is the cost of applying update
schema $u$ to the $\stmt(\ell)$.  We impose an obvious restriction on $\cost$:
$\forall \ell \in \cL: \cost(\id,\ell) = 0$.  Since we have already fixed the
set $\U$ and the set $\cL$ of locations of program $\cP$ (or equivalently, of
Boolean program $\cB$), we henceforth use $c$, $\upd$ instead of $\cost$,
$\update$, respectively, The total cost, $\Cost_{c}(\upd)$, of performing a
$\upd$-update of $\cP$ is given by $\sum_{\ell \in \cL} \, c(\upd(\ell),\ell)$.

Given an incorrect
concrete program $\cP$ annotated with assertions,
a cost function $c$ and a repair budget
$\target$, the goal of cost-aware program repair is to compute $\wcP$
such that:

\begin{enumerate}
\item $\wcP$ is totally correct, and,
\item there exists $\upd$:
\begin{enumerate}
\item $\wcP$ is some $\upd$-update of $\cP$, and
\item $\Cost_{c}(\upd) \leq \target$.
\end{enumerate}
\end{enumerate}

\noindent If there exists such a $\wcP$, we say $\wcP$ is a {\em
$(\U,c,\target)$-repair of $\cP$}.

In addition to the above problem, we propose another problem as
follows. Let $\cT = \{\cT_1,\ldots,\cT_h\}$ be a set of {\em
templates} or {\em grammars}, each representing a syntactical restriction for the
modified expressions in $\wcP$.  The syntax of an example template,
say $\cT_1$, defining Boolean-valued linear arithmetic expressions
over the program variables, denoted $\en{blaexpr}$, is shown below:
\begin{center}
\begin{tabular}[h]{l l l}
$\en{blaexpr}$ & $::=$ & $atom$ $|$ $(\,\en{blaexpr}\,)$ $|$ $\neg \en{blaexpr}$
$|$ $\en{blaexpr} \wedge \en{blaexpr}$\\
$\en{atom}$ & $::=$ & $\en{laterm} \, \en{cmp} \, \en{laterm}$\\
$\en{laterm}$ & $::=$ & $const$ $|$ $var$ $|$ $const \times
var$
$|$ $\en{laterm} + \en{laterm}$\\
$\en{cmp}$ & $::=$ & $=$ $|$ $<$ $|$ $\leq$.\\
\end{tabular}
\end{center}
\noindent In the above, $const$ and $var$ denote integer-valued or real-valued
constants and program variables, respectively. Expressions such as
$\varname_1 +
2 \times \varname_2 \leq \varname_3$, $(\varname_1 < \varname_2) \, \wedge \,
(\varname_3 = 3)$ \etc,
that satisfy the syntactical requirements of the template $\cT_1$,
are said to belong to the {\em language} of the template, denoted
$L(\cT_1)$.

Let $\temp: \cL \to \cT$, be a function that maps each location of
$\cP$ to a template in $\cT$.  Let $\cE(\stmt(\ell))$ denote a set
that includes all expressions in certain statement types
and be defined as follows: if $\stmt(\ell)$ is
$\varname_1,\ldots,\varname_\assind  \; \bm{\assign} \; e_1,\ldots,
e_\assind$, $\cE(\stmt(\ell)) = \{e_1,\ldots,e_\assind\}$, else if
$\stmt(\ell)$ is $\mtt{call} \, F_j(e_1,\ldots,e_\formalind)$,
$\cE(\stmt(\ell)) = \{e_1,\ldots,e_\formalind\}$, else if
$\stmt(\ell)$ is $\mtt{assume} \, (g)$, $\cE(\stmt(\ell)) = \{g\}$
else, $\cE(\stmt(\ell))$ is the empty set.

Given $\temp$, along with (incorrect) $\cP$, 
$c$ and $\target$, the goal of {\em template-based}, cost-aware
program repair is to compute $\wcP$ such that:

\begin{enumerate}
\item $\wcP$ is correct, and,
\item there exists $\upd$:
\begin{enumerate}
\item $\wcP$ is some $\upd$-update of $\cP$,
\item $\Cost_c(\upd) \leq \target$, and
\item for each location $\ell$:\\
$\upd(\ell) \neq \id  \; \Rightarrow \forall e \in
\cE(\wstmt(\ell)): \, e \in L(\temp(\ell))$.
\end{enumerate}
\end{enumerate}

We conjecture that an insightful choice for the cost function $c$ can
help prune the search space for repaired programs and help incorporate
expert user intuition and intent in automatic program repair.
Exploration of suitable cost-functions is beyond the scope of this
dissertation. For now, we would only like to emphasize that our
cost-function is quite flexible, and can be used to constrain the
computation of $\wcP$ in diverse ways. For example, the user can
choose to search for $\wcP$ that differs from $\cP$ in at most
$\target$ statements by defining $c$ as:
\begin{align*}
\forall \ell \in \cL, u
\in \U: u \neq \id \; \Rightarrow \; \cost((\id,\ell)) = 1.
\end{align*}
Or, the user can choose to search for $\wcP$ that does not modify any
statement within a {\em trusted} procedure $\cP_i$ by defining $c$
as:
\begin{align*}
\forall \ell \in \cL, u \in \U: \, &u\neq id \, \wedge \, \ell \in
\cL_i \;
\Rightarrow \; \cost((u,\ell)) = N \text{  and } \\
&u \neq id \, \wedge \, \ell
\not\in \cL_i  \; \Rightarrow \; \cost((u,\ell)) = 1,
\end{align*}
where $N$ is some prohibitively large number. Or, the user can choose
to favor the application of a particular update schema, say $u_1$,
over
others by defining $c$ as:
\begin{align*}
\forall \ell \in \cL, u
\in \U:\, &u \neq id \text{ and } u \neq u_1 \; \Rightarrow \;
\cost((u,\ell))
= N \text{  and}\\
&u = u_1  \; \Rightarrow \; \cost((u,\ell)) = 1,
\end{align*}
where $N$ is some prohibitively large number, and so on.
Similarly, insightful templates choices can help guide the search for
repairs based on user input.

\subsection{Solution Overview}

We present a predicate abstraction-based framework for cost-aware
program repair.  Recall that we had fixed a Boolean program $\cB$ in
\secref{prelims} such that $\cB$ is obtained from $\cP$ via iterative
predicate abstraction-refinement and $\cB$ exhibits a non-spurious
counterexample path. In addition to $\cP$, $\cost$ and
$\target$, our framework requires: the Boolean program $\cB$ and the
corresponding function $\conc$ that maps Boolean variables to their
respective predicates.  The computation of a suitable repaired program
$\wcP$ involves two main steps:

\begin{enumerate}
\item Cost-aware repair of $\cB$ to obtain $\wcB$, and 
\item Concretization of $\wcB$ to obtain $\wcP$.
\end{enumerate}

The problem of cost-aware repair of a Boolean program $\cB$ can be
defined in a manner identical to cost-aware repair of concrete program
$\cP$. Concretization of $\wcB$ involves mapping each statement of
$\wcB$ that has been modified by $\update$ into a corresponding
statement of $\cP$, using the function $\conc$. For template-based
repair of $\cP$, the concretization needs to ensure that the modified
expressions of $\cP$ meet the syntactic requirements of the
corresponding templates. In the following sections, we describe these
two steps in detail.

\section{Cost-aware Repair of Boolean Programs}\label{sec:algo1}

\noindent Our solution to cost-aware repair of a Boolean program $\cB$ relies on
automatically computing {\em inductive assertions}, along with a
suitable $\wcB$, that together certify the partial correctness of
$\wcB$. In what follows, we explain our adaptation of the method
of inductive assertions \cite{Floyd67,Manna74}
for cost-aware program repair. \\

\noindent{\bf Cut-set}. Let $\N = \N_0 \cup \ldots \cup \N_\procind$
be the set of nodes in $\cG(\cB)$, the transition graph representation
of $\cB$. We define a cut-set $\Lambda \subseteq \N$ as a set of
nodes, called {\em cut-points}, such that for every $i \in [0,\procind]$:
(a) $\entry_i, \exit_i \in \Lambda$, (b) for every edge $(\ell,
\st,\ell') \in \R_i$ where $\st$ is a procedure $\mtt{call}$, $\ell,
\ell' \in \Lambda$, (c) for every edge $(\ell, \st,\ell') \in \R_i$
where $\st$ is an $\mtt{assert}$ statement, $\ell,\ell' \in \Lambda$,
and (d) every cycle in $\cG_i$ contains at least one node in
$\Lambda$. A pair of cut-points $\cs$, $\ce$ in some $\cG_i$ is
said to be {\em adjacent} if every path from $\cs$ to $\ce$ in
$\cG_i$ contains no other cut-point. A {\em verification path} is any
path from a cut-point to an adjacent cut-point; note that there can be
more than one verification path between two adjacent cut-points.  \\

\noindent {\em Example}: The set $\{\ell_1, \ell_2, \ell_7, exit\}$ is a valid
cut-set for Boolean program $\cB$ in \figref{runexPBTG}.
The verification paths in $\cG(\cB)$ corresponding to this cut-set
are as follows:
\begin{enumerate}
\item $\ell_1 \tran{\mtt{assume} \, (b_2)} \ell_2$
\item $\ell_2 \tran{\mtt{assume} \, (\mtt{T})} \ell_3 \tran{b_0, b_1,
b_2  \, \assign \, *, *, *} \ell_4 \tran{}  \ell_2$
\item $\ell_2 \tran{\mtt{assume} \, (\mtt{T})} \ell_0 \tran{} \ell_7$
\item $\ell_1 \tran{\mtt{assume} \, (\neg b_2)} \ell_5 \tran{\mtt{assume}
\, (\neg b_1)} \ell_7$
\item $\ell_1 \tran{\mtt{assume} \, (\neg b_2)} \ell_5 \tran{\mtt{assume}
\, (b_1)} \ell_6 \tran{b_1, b_1, b_2 \, \assign \, *, *, *} \ell_7$
\item $\ell_7 \tran{\mtt{assert} \, (\neg b_0)} exit$\footnote{Labeling
this edge with $\mtt{assert} \, (\neg b_0)$ is a slight abuse of the
semantics of an assert statement. Our justification
is that the constraints formulated later in this section require that
the assertion is $\true$ whenever control reaches location $\ell_7$ in
an execution path.}
\end{enumerate}

\smallskip

%For example, in \figref{example}, the set %of shaded nodes $\{ \}$ is a
%possible cut-set, and there are $\#$ %verification paths.

\noindent{\bf Inductive assertions}. We denote an inductive assertion
associated with cut-point $\cut$ in $\Lambda$ by $\I_\cut$.
Informally, an inductive assertion $\I_\cut$ has the property that
whenever control reaches $\cut$ in any program execution, $\I_\cut$
must be $\true$ for the current values of the variables in scope.
Thus, for a Boolean program, an inductive assertion $\I_\cut$ is in
general a Boolean formula over the variables whose scope includes
$\cut$. To be precise, $\I_\cut$ is a Boolean formula over
$\Var_s[\cut]$, where $\Var_s[\cut]$ denotes an $\cut^{th}$ copy of
the subset $\Var_s$ of the program variables, with $\Var_s = GV \,
\cup \, formal(\cut)$ if $\cut \in \{\exit_1,\ldots,\exit_\procind\}$,
and $\Var_s = \inscope(\cut)$ otherwise.  Thus, except for the
$\mtt{main}$ procedure, the inductive assertions at the exit nodes of
all procedures exclude the local variables declared in the procedure.
Let $\I_\Lambda$ denote the set of inductive assertions associated
with all the cut-points in $\Lambda$. \\

\noindent{\bf Verification conditions}. A popular approach to
verification of sequential, imperative programs is to compute
$\I_\Lambda$ such that $\I_\Lambda$ satisfies a set of constraints
called {\em verification conditions}. Let $\pi$ be a verification path
in $\cG_i$, from cut-point $\cs$ to adjacent cut-point $\ce$.  The
verification condition corresponding to $\pi$, denoted $\VC(\pi)$, is
essentially the Hoare triple $\en{\I_\cs} \, \stmt(\pi) \,
\en{\I_\ce}$, where $\stmt(\pi)$ is the sequence of statements
labeling $\pi$. When $\I_\cs$, $\I_\ce$ are {\em unknown}, 
$\VC(\pi)$ can be seen as a constraint encoding all possible solutions for 
$\I_\cs$, $\I_\ce$ such that: every program execution
along path $\pi$, starting from a set of variable valuations
satisfying $\I_\cs$, terminates in a set of variable valuations
satisfying $\I_\ce$. Note that the definitions of cut-sets and
adjacent cut-points ensure that we do not have to worry about
non-termination along verification paths. 

The Hoare triple $\en{\I_\cs} \, \stmt(\pi) \, \en{\I_\ce}$ can be
defined using weakest preconditions or strongest postconditions. In
this paper, as we will see shortly, we find it convenient to use
strongest postconditions.\\

\noindent{\bf Program verification using the inductive assertions
method}. Given a program $\cB$ annotated with assertions, and a set
$\Lambda$ of cut-points, $\cB$ is partially correct if one can compute
a set $\I_\Lambda$ of inductive assertions such that: for every
verification path $\pi$ between every pair $\cs,\ce$ of adjacent
cut-points in $\cG$, $\VC(\pi)$ is valid.\\

\noindent{\bf Cost-aware repairability conditions for partial correctness}.
Let $\C: \bigcup_{i=0}^t \N_i \to \Nat$ be a function mapping locations to
costs.  We find it convenient to use $\C_\cs$ to denote the value $\C(\cs)$ at
location $\cs$. We set $\cI_{\entry_0} = \true$ and $\C_\cut = 0$ if $\cut \in
\{\entry_0,\ldots, \entry_\procind\}$.  Informally, $\C_\cs$ with $\cs \in
\N_i$ can be seen as recording the cumulative cost of applying a sequence of
update schemas to the statements in procedure $F_i$ from location $\entry_i$ to
$\cs$.  Thus, for a specific update function $\upd$ with cost function $c$,
$\C_{\exit_0}$ records the total cost $Cost_c(\upd)$ of performing an
$\upd$-update of the program.  Given a verification path $\pi$ in $\cG_i$, from
cut-point $\cs$ to adjacent cut-point $\ce$, we extend the definition of
$\VC(\pi)$ to define the cost-aware repairability condition corresponding to
$\pi$,  denoted $\CRC(\pi)$. $\CRC(\pi)$ can be seen as a constraint encoding
all possible solutions for inductive assertions $\cI_\cs$, $\cI_\ce$ and update
functions $\update$, along with associated functions $\C$, such that: every
program execution that proceeds along path $\pi$ via statements modified by
applying the update schemas in $\update$, starting from a set of variable
valuations satisfying $\cI_\cs$, terminates in a set of variable valuations
satisfying $\cI_\ce$, for all nondeterministic choices that the program might
make along $\pi$.

Before we proceed, recall that $\cI_\cut$ is a Boolean formula over
$\Var_s[\cut]$, with $\Var_s = \gVar \, \cup \, formal(\cut)$ if $\cut
\in \{\exit_1,\ldots,\exit_\procind\}$, and $\Var_s = \inscope(\cut)$
otherwise.  Thus, for all locations $\ci \neq \ce$ in verification
path $\pi$ from $\cs$ to $\ce$, $\Var_s = inscope(\ci)$.  In what
follows, the notation $\funcDefinedBy{u}(\stmt(\ci))$ represents the
class of statements that may be obtained by applying update schema $u$
on $\stmt(\ci)$, and is defined for our permissible update schemas in
\figref{updef}. Here, $f, f_1, f_2$ \etc denote unknown
Boolean expressions\footnote{To keep our exposition simple, we
assume that these unknown Boolean expressions are deterministic.
However, in our prototype tool (see \secref{expt}), we also
have the ability to compute modified statements with
nondeterministic expressions such as $*$ or $\mtt{choose}(f_1,f_2)$.},
over the variables in $inscope(\lambda)$. Note that the update schema
$\mtt{assign} \mapsto \mtt{assign}$, modifies {\em any} assignment
statement, to one that assigns unknown Boolean expressions to {\em
all} variables in $\Var_s$.

\begin{figure}
\begin{center}
\begin{tabular}[h]{|l l|l l|}
\hline
$u$ & & $\funcDefinedBy{u}(\stmt(\ci))$ & \\
\hline
$id$ & & $\stmt(\ci)$ &\\
$\mtt{assign} \mapsto \mtt{skip}$ & & $\mtt{skip}$ & \\
$\mtt{assume} \mapsto \mtt{skip}$ & & $\mtt{skip}$ & \\
$\mtt{call} \mapsto \mtt{skip}$ & & $\mtt{skip}$ & \\
$\mtt{assign} \mapsto \mtt{assign}$ & & $\var_1,\ldots,\var_{|\Var_s|} \;
\bm{\assign} \; f_1,\ldots,f_{|\Var_s|}$ & \\
$\mtt{assume} \mapsto \mtt{assume}$ & &$\mtt{assume} \, f$ & \\
$\mtt{call} \mapsto \mtt{call}$ & &$\mtt{call} \; F_j
(f_1,\ldots,f_\formalind)$, where $\stmt(\ci)$: $\mtt{call} \;
F_j(e_1,\ldots,e_\formalind)$ &\\
\hline
\end{tabular}
\end{center}
\caption{Definition of $\funcDefinedBy{u}(\stmt(\ci))$}
\label{fig:updef}
\end{figure}

\noindent We now define $\CRC(\pi)$. There are three cases to consider.

\begin{enumerate}

\item $\stmt(\pi)$ does not contain a procedure $\mtt{call}$
or $\mtt{assert}$ statement: 

Let $\A_\ci$ denote an assertion associated
with location $\ci$ in $\pi$. $\CRC(\pi)$ is
given by the (conjunction of the) following set of constraints: 
\begin{align}
\label{eqcrc}
%\label{crc}
\A_\cs \; = \; &\I_\cs \notag\\
%\Q[\cs] \; = \; &\C_\cs\\
\A_\ce \; \Rightarrow \; &\I_\ce\\
%\Q[\ce] \; = \; &\C_\ce\\
\bigwedge_{\cs \preceq \ci \prec \ce} \, 
\bigwedge_{u \in \U_{\stmt(\ci)}} \; 
\upd(\ci) = u  \;\; \Rightarrow \; \; &\C_\cj = \C_\ci +  c(u,\ci)
\; \wedge \notag\\
                               &\A_\cj =
                               sp(\funcDefinedBy{u}(\stmt(\ci)),\A_\ci)\notag.
\end{align}
%\end{center}

\noindent In the above, $\prec$ denotes the natural ordering over the
sequence of locations in $\pi$ with $\ci$, $\cj$ being consecutive
locations, \ie, $\cj \in succ(\ci)$.  The notation $\U_{\stmt(\ci)}
\subseteq \U$ denotes the set of all update schemas in $\U$ which may
be applied to $\stmt(\ci)$. The notation
$sp(\funcDefinedBy{u}(\stmt(\ci)),\A_\ci)$ denotes the strongest
postcondition of the assertion $\A_\ci$ over the class of statements
$\funcDefinedBy{u}(\stmt(\ci))$. We define the strongest postcondition
using multiple variable copies - a copy $\Var_s[\ci]$ for each
location $\ci$ in $\pi$.

Let us assume that $\A_\ci$ is a Boolean formula of the
form\footnote{In general, $\A_\ci$ is a
disjunction over Boolean formulas of this form;
$sp(\funcDefinedBy{u}(\stmt(\ci)),\A_\ci)$ can then be obtained by computing a
disjunction over the strongest postconditions obtained by propagating
each such Boolean formula through $\funcDefinedBy{u}(\stmt(\ci)$
using the rules in \figref{spdef}.}:
\begin{align}\label{eqspform} 
\A_\ci = \pc[\cut,\ch]
\, \wedge \bigwedge_{\var \in \Var_s} \var[\ci] = \E[\ch],
\end{align} 

\noindent where $\ch$, $\ci$ are consecutive locations in $\pi$ with $\ci \in
succ(\ch)$, $\pc[\cut,\ch]$ is a Boolean expression over all copies
$\Var_s[\mu]$, $\cut \preceq \mu \preceq \ch$, representing the path condition
imposed by the program control-flow, and $\E[\cut]$ is a Boolean expression over
$\Var_s[\ch]$ representing the $\ci^{th}$ copy of each variable $\var$ in terms
of the $\ch^{th}$ copy of the program variables. Note that $\A_\cut = \I_\cut$
is of the form $\pc[\cut]$. 

\begin{figure}
\centering
\begin{tabular}[h]{|l | p{6.5cm} |}
\hline
$\funcDefinedBy{u}(\stmt(\ci))$ & $sp(\funcDefinedBy{u}(\stmt(\ci)),\A_\ci)$ \\
\hline
$\mtt{skip}$ & \multirow{2}{*}{$\pc[\cut,\ch] \, \wedge \, \bigwedge_{\var
\in \Var_s} \var[\cj] = \var[\ci]$}\\
$\mtt{goto}$ & \\
\hline
\multirow{2}{*}{$\mtt{assume} \, g$} & \multirow{2}{*}{$g[\ci] \, \wedge \,
\pc[\cut,\ch] \, \wedge \, \bigwedge_{\var \in \Var_s} \var[\cj] =
\var[\ci]$}\\%[3ex]
& \\
\hline
\multirow{2}{*}{$\mtt{assume} \, f$} & \multirow{2}{*}{$f[\ci] \, \wedge \,
\pc[\cut,\ch] \, \wedge \, \bigwedge_{\var \in \Var_s} \var[\cj] =
\var[\ci]$}\\%[3ex]
& \\
\hline
\multirow{3}{*}{$\var_1,\ldots,\var_\assind \; \bm{\assign} \;
e_1,\ldots,e_\assind$} & \\
& $\pc[\cut,\ch] \, \wedge \, \bigwedge_{\var_i \in \Var_s, i \in [1,m]}
\var_i[\cj] = e_i[\ci] \;  \wedge $\\
%&\\
& \hspace{1.3cm}$\bigwedge_{\var_i \in \Var_s, i \not\in [1,m]} \var_i[\cj] =
\var_i[\ci]$\\%[4ex]
& \\
\hline
\multirow{3}{*}{$\var_1,\ldots,\var_{|\Var_s|} \; \bm{\assign} \;
f_1,\ldots,f_{|\Var_s|}$} & \\
& $\pc[\cut,\ch] \, \wedge \, \bigwedge_{\var_i \in \Var_s}
\var_i[\cj] = f_i[\ci]$\\
& \\
\hline
\end{tabular}
\caption{Definition of $sp(\funcDefinedBy{u}(\stmt(\ci)),\A_\ci)$}
\label{fig:spdef}
\end{figure}

Given $\A_\ci$ of the form in (\ref{eqspform}),
$sp(\funcDefinedBy{u}(\stmt(\ci)),\A_\ci)$ is defined in \figref{spdef}.
Observe that $sp(\funcDefinedBy{u}(\stmt(\ci)),\A_\ci)$ is a Boolean formula of
the same form as (\ref{eqspform}), over variable copies from $\Var_s[\cut]$ to
$\Var_s[\cj]$. For the entries $\mtt{assume} \, g$ and
$\var_1,\ldots,\var_\assind \; \bm{\assign} \; e_1,\ldots,e_\assind$, the
expressions $g,e_1,\ldots,e_\assind$ are {\em known} beforehand (these entries
correspond to $u = id$). For the entries $\mtt{assume} \, f$ and
$\var_1,\ldots,\var_{|\Var_s|} \; \bm{\assign} \; f_1,\ldots,f_{|\Var_s|}$, the
expressions $f$, $f_1$, $\ldots$, $f_{|\Var_s|}$ are {\em unknown} (these
entries correspond to $u = \mtt{assume} \mapsto \mtt{assume}$ and $u =
\mtt{assign} \mapsto \mtt{assign}$, respectively). Notation such as $f[\ci]$
denotes that $f$ is an unknown Boolean expression over $\Var_s[\ci]$.  For
nondeterministic expressions in the RHS of an assignment statement
$\var_1,\ldots,\var_\assind \; \bm{\assign} \; e_1,\ldots,e_\assind$, the
strongest postcondition is computed as the disjunction of the strongest
postconditions over all possible assignment statements obtained by substituting
each $*$ expression with either $\false$ or $\true$.

Thus, to summarize, the set of constraints in $(\ref{eqcrc}$) encodes
all $\I_\cs$, $\C_\cs$, $\I_\ce$, $\C_\ce$ and $\update$ such that: if
$\update$ is applied to the sequence of statements $\stmt(\pi)$ to get
some modified sequence of statements, say $\widehat{\stmt}(\pi)$, and
program execution proceeds along $\widehat{\stmt}(\pi)$, then $\I_\ce$
is the strongest postcondition $sp(\widehat{\stmt}(\pi),
\I_\cs)$, and $\C_\ce$ equals the cumulative modification cost,
counting up from $\C_\cs$. 

\item $\stmt(\pi)$ contains a procedure $\mtt{call}$, say
$\mtt{call} \; F_j(e_1,\ldots,e_\formalind)$:

The path $\pi$, given by $(\cs,\mtt{call} \;
F_j(e_1,\ldots,e_\formalind),\ce)$,
is a verification path of length $1$.
Suppose the formal
parameters of $F_j$ are $\var_1,\ldots,\var_\formalind$.\\
$\CRC(\pi)$ is then given by the following set of constraints:
\begin{align}
\label{crcproc}
\upd = id \;\; \Rightarrow \; 
&\C_\ce = \C_\cs +  \C_{\exit_j} \;\, \wedge \notag\\
&\cI_\cs \; \Rightarrow \; \cI_{\entry_j}[\bigwedge_{i\in[1,\formalind]}
\var_i[\entry_j]/e_i[\cs]] \;\, \wedge \notag\\
&\cI_{\exit_j}[\bigwedge_{i\in[1,\formalind]}\var_i[\exit_j]/e_i[\ce]] \; \Rightarrow \; \cI_\ce \notag\\
\upd = \mtt{call} \mapsto \mtt{skip} \;\; \Rightarrow \; 
&\C_\ce = \C_\cs +  c(\mtt{call} \mapsto \mtt{skip},\cs) \;\, \wedge \\
&\cI_\ce = \cI_\cs[\bigwedge_{i\in[1,\formalind]}\var_i[\cs]/\var_i[\ce]] \notag\\
\upd = \mtt{call} \mapsto \mtt{call} \;\; \Rightarrow \; 
&\C_\ce = \C_\cs + \C_{\exit_j} + c(\mtt{call} \mapsto \mtt{call},\cs) \;\, \wedge \notag\\ 
&\cI_\cs \; \Rightarrow \;
\cI_{\entry_j}[\bigwedge_{i\in[1,\formalind]}\var_i[\entry_j]/f_i[\cs]] \; \wedge \notag\\
&\cI_{\exit_j}[\bigwedge_{i\in[1,\formalind]}\var_i[\exit_j]/f_i[\ce]] \; \Rightarrow \; \cI_\ce \notag 
\end{align}

For $\upd = id$, the constraints involve replacing the
$\entry_j^{th}$, $\exit_j^{th}$ copies of the formal parameters in
$\cI_{\entry_j}$, $\cI_{\exit_j}$ with the corresponding actual
parameters $e_1,\ldots,e_\formalind$ expressed over the $\cs^{th}$,
$\ce^{th}$ copies of the program variables, respectively. For $\upd =
\mtt{call} \mapsto \mtt{call}$, a similar substitution is performed,
except the actual parameters are unknown expressions
$f_1,\ldots,f_\formalind$.  Finally, for $\upd = \mtt{call} \mapsto
\mtt{skip}$, the inductive assertion essentially stays the same, with
variable copies appropriately adjusted. $\C_\ce$ is in general the sum
of $\C_\cs$, the cumulative modification cost $\C_{\exit_j}$ of
procedure $F_j$, and the cost of applying the update schema in
question.

\item $\stmt(\pi)$ contains an $\mtt{assert}$ statement, say
$\mtt{assert} \, g$:

Again, $\pi$, given by $(\cs,\mtt{assert} \, g,\ce)$, is
a verification path of length $1$, and $\CRC(\pi)$ is given by the
following set of constraints:
\begin{align*}
\cI_\cs[\bigwedge_{i\in[1,|\Var_s|]} \var_i[\cs]/\var_i[tmp]]  \;
\Rightarrow \; &g[\bigwedge_{i\in[1,|\Var_s|]}\var_i/\var_i[tmp]]\\ 
\cI_\cs[\bigwedge_{i\in[1,|\Var_s|]} \var_i[\cs]/\var_i[tmp]]  \;
\Rightarrow \; 
&\cI_\ce[\bigwedge_{i\in[1,|\Var_s|]}\var_i[\ce]/\var_i[tmp]]\\ 
\C_\ce \; = \; &\C_\cs.
\end{align*}
In the above, we uniformly convert the expressions $\cI_\cs$, $g$ and $\cI_\ce$
into expressions over some temporary copy of the program variables to
enable checking the implications (informally, these implications are
$\cI_\cs \Rightarrow g$ and $\cI_\cs \Rightarrow \cI_\ce$).

\end{enumerate}

\noindent{\bf Cost-aware Boolean program repair}.
Given a cut-set $\Lambda$ of $\cG(\cB)$, let $\Pi_\Lambda$ be the set of all
verification paths between every pair of adjacent cut-points in
$\Lambda$. Given incorrect program $\cB$ annotated with assertions,
the set $\cU$, cost function $c$ and repair budget $\target$,
we say $\cB$ is {\em repairable within
budget $\target$} if given a cut-set $\Lambda$ in $\cG$,
one can compute a set $\cI_\Lambda$ of inductive assertions, an update
function $\upd$, along with models for all unknown
expressions associated with applications of update schemas in $\upd$,
and the valuations of a cumulative-cost-recording function $\C$ such
that: $\C_{exit_0} \leq \target$, for every verification path $\pi
\in \Pi_\Lambda$, $\CRC(\pi)$ is valid and some other constraints are
met. Mathematically, $\cB$ is repairable within
budget $\target$ if the following formula is $\true$:
\begin{align}
\label{PartCorrect}
\exists Unknown \; \forall Var: \;\; \C_{exit_0} \leq \target \;
\wedge \; \bigwedge_{\pi \in \Pi_\Lambda} \CRC(\pi) \; \wedge \;
AssumeConstraints, 
\end{align}
where $Unknown$ is the set of all unknowns and $Var$ is the set of all Boolean
program variables and their copies used in encoding each constraint $\CRC(\pi)$. The set of unknowns includes the inductive assertions in $\cI_\Lambda$,
update function $\upd$, unknown expressions $f, f_1$ \etc associated
with applying the update schemas in $\upd$ and valuations at each
program location of the function $\C$. Finally, $AssumeConstraints$
ensures that any modifications to the guards of
$\mtt{assume}$ statements corresponding to the same
conditional statement are consistent. Thus, for every pair of
{\em updated} $\mtt{assume} \, (f_1)$, $\mtt{assume} \, (f_2)$ statements
labeling edges starting from the same
node in the transition graph, the uninterpreted
functions $f_1$, $f_2$ are constrained to satisfy $f_1 = \neg f_2$.

If the above formula is $\true$, then we can extract models for all the
unknowns from the witness to the satisfiability of the formula: $\forall Var:$
$\C_{exit_0} \leq \target$ $\wedge$ $\bigwedge_{\pi \in \Pi_\Lambda}
\CRC(\pi)$ $\wedge$ $AssumeConstraints$.  In
particular, we can extract an $\upd$ and the corresponding modified statements
to yield a correct Boolean program $\wcB$. The following theorem
states the correctness and completeness of the above algorithm for repairing
Boolean programs for partial correctness.

\begin{thm}
Given the set $\U$ specified in (\ref{upschema}),
and given an incorrect Boolean program $\cB$ annotated with assertions,
cost function $c$ and repair budget $\target$,
\begin{enumerate}
\item if there exists a $(\U,c,\target)$-repair of $\cB$,
the above method finds a $(\U,c,\target)$-repair of $\cB$,
\item if the above method finds a $\wcB$, then $\wcB$ is a $(\U,c,\target)$-repair of $\cB$.
\end{enumerate}
\end{thm}

\begin{proof}
Note that the formula (\ref{PartCorrect}) is a $\exists\forall$ formula over
Boolean variables (Boolean program variables and their copies), unknown Boolean
expressions over these Boolean variables (inductive assertions and expressions
in modified program statements), sequences of update schemas (update
functions) and corresponding sequences of integer costs (valuations of
$\C$). The
number of Boolean variables is finite and hence, the number of unknown Boolean
expressions over them is finite. There are a finite number of update functions
drawn from finite sequences of update schemas in the finite set $\cU$, and a
corresponding finite number of $\C$ functions, with $\C_{\entry_0}$ set to $0$.
Besides these (\ref{PartCorrect}) includes Boolean operators, the $+$ operator and
a finite number of integer constants (corresponding to the cost function $c$).
Clearly, the truth of the formula in (\ref{PartCorrect}) is decidable. In particular,
the formula has a finite number of models.

Given the set $\U$ specified in (\ref{upschema}), the completeness of our method
follows from the completeness of Floyd's inductive assertions method
and the decidability of the formula in (\ref{PartCorrect}).

The soundness of our method follows from the soundness of Floyd's inductive
assertions method.
\end{proof}

\noindent {\em Example:} For the Boolean program in \figref{runex.B}, our tool
modifies two statements: ($1$) the guard for $stmt(\ell_1)$ is changed from
$b2$ to $b0 \vee b1 \vee \neg b2$ and ($2$) the guard for $stmt(\ell_2)$ is
changed from $*$ to $b0 \vee b1 \vee b2$.

\section{Concretization}\label{sec:algo2}

We now present the second step in our framework for computing a concrete
repaired program $\wcP$. In what follows, we assume that we have already
extracted models for $\wcB$ and  $\I_\Lambda$.
Recall that $\conc$ denotes the mapping of Boolean variables to their
respective predicates: for each $i \in [1,|\Var(\cB)|]$, $\conc(\var_i) =
\pred_i$.  The mapping $\conc$ can be extended in a standard way to map
expressions over the Boolean variables in $\Var(\cB)$ to expressions over the
concrete program variables in $\Var(\cP)$. \\ 

\noindent{\bf Concretization of $\wcB$}. The goal of concretization of a
repaired Boolean program $\wcB$ is to compute a corresponding repaired concrete
program $\wcP$. This involves computing a mapping, denoted $\Conc$, from each
modified statement of $\wcB$ into a corresponding modified statement in the
concrete program.  In what follows, we define $\Conc$ for each type of modified
statement in $\wcB$.  Let us fix our attention on a statement at location
$\ell$, with $\Var_s(\cB)$, $\Var_s(\cP)$ denoting the set of concrete,
abstract program variables, respectively, whose scope includes $\ell$.  Let
$\predind = |\Var_s(\cB)|$ and $q = |\Var_s(\cP)|$.

\begin{enumerate} 
\item $\Conc(\mtt{skip}) = \mtt{skip}$
\item $\Conc(\mtt{assume}\, (g)) = \mtt{assume} \, (\conc(g))$
\item $\Conc(\mtt{call} \; F_j(e_1,\ldots,e_\formalind) = \mtt{call} \; F_j(\conc(e)_1,\ldots,\conc(e)_\formalind)$
\item The definition of $\Conc$ for an assignment statement is non-trivial. In fact, in this case, $\Conc$ may be the empty set, 
or may contain multiple concrete assignment statements. 

We say that an assignment statement $\var_1,\ldots,\var_{\predind} \;
\bm{\assign} \; e_1,\ldots,e_\predind$ in $\cB$ is {\em concretizable} if one
can compute expressions $f_1,\ldots,f_q$ over $\Var_s(\cP)$, 
of the same type as the concrete program variables 
$\varname_1, \ldots, \varname_q$ in $\Var_s(\cP)$, respectively, such that a certain set of constraints
is valid. To be precise, $\var_1,\ldots,\var_{\predind} \;
\bm{\assign} \; e_1,\ldots,e_\predind$ in $\cB$ is concretizable if the following formula 
is $\true$:
\begin{align}
\label{conc}
\exists f_1,\ldots, f_q \; \forall \varname_1,\ldots, \varname_q: \;\; 
\bigwedge_{i=1}^\predind \; \conc(\var_i)[\varname_1/f_1,\ldots,\varname_{q}/f_{q}] \; = \; \conc(e_i)
%&\wedge \notag\\
%&\ldots\\
%&\wedge \notag \\
%\conc(\var_\predind)[\varname_1/f_1,\ldots,\varname_{q}/f_q] \; &= \; \conc(e_\predind).\notag
\end{align}
Each quantifier-free constraint $\conc(\var_i)[\varname_1/f_1,\ldots,\varname_{q}/f_{q}] \; = \; \conc(e_i)$ above
essentially expresses the concretization of the abstract assignment $\var_i = e_i$. The substitutions 
$\varname_1/f_1,\ldots,\varname_{q}/f_{q}$ reflect the {\em new} values of the concrete program variables 
after the concrete assignment $\varname_1, \ldots, \varname_{q}$ $\bm{\assign}$ $f_1,\ldots,f_q$. 

If the above formula is $\true$, we can extract models $expr_1,\ldots,expr_q$ for $f_1,\ldots, f_q$, respectively, 
from the witness to the satisfiability of the inner $\forall$-formula. We then say:
\begin{align*}
\varname_1, \ldots, \varname_{q} \; \bm{\assign} \; expr_1,\ldots,expr_q \;
\in \;  \Conc(\var_1,\ldots,\var_{\predind} \; \bm{\assign} \;
e_1,\ldots,e_\predind).
\end{align*}

Note that, in practice, for some $i \in [1,q]$, $expr_i$ may be
equivalent to $\varname_i$, thereby generating a redundant assignment
$\varname_i \; \assign \; \varname_i$. The parallel
assignment can then be compressed by eliminating each redundant
assignment. In fact, it may be possible to infer some such $\varname_i$ without using
(\ref{conc}) by analyzing the dependencies of concrete program variables on the
predicates in $\{\phi_1,\ldots, \phi_\predind\}$ that are actually affected by the
Boolean assignment in question; this exercise is beyond the
current scope of this work.
\end{enumerate}

\noindent{\bf Template-based concretization of $\wcB$}.  Recall that
$\temp(\ell)$, associated with location $\ell$, denotes a user-supplied
template from $\cT$, specifying the desired syntax of the expressions in any
concrete modified statement at $\ell$. Henceforth, we use the shorthand $\tel$ for
$\temp(\ell)$. We find it helpful to illustrate template-based concretization
using an example template. Let us assume that for each concrete program variable
$\varname \in \Var(\cP)$, $\varname \in \Nat \cup \Reals$.  We fix $\tel$ to 
(Boolean-valued) linear arithmetic expressions over the program variables, of
the form $c_0 + \Sigma_{p=1}^q c_p*\varname_p \leq 0$, for $\mtt{assume}$ and
$\mtt{call}$ statements, and (integer or real-valued) linear arithmetic {\em
terms} over the program variables, of the form $c_0 + \Sigma_{p=1}^q c_p*\varname_p$,
for assignment statements. Let us assume that the parameters $c_0, c_1,\ldots,c_q
\in \Reals$. Given $\tel$, let $\Conc_{\tel}$ denote the mapping of abstract statements into 
concrete statements compatible with $\tel$. We can define $\Conc_{\tel}$ for each type of modified
statement in $\wcB$ as shown below. The basic idea is to compute suitable values
for the template parameters $c_0,\ldots,c_q$ that satisfy
certain constraints.  Note that, in general, $\Conc_{\tel}$ may be the empty set,
or may contain multiple concrete statements.

\begin{enumerate}
\item $\Conc_{\tel}(\mtt{skip}) = \mtt{skip}$
\item The statement $\mtt{assume}\, (g)$ is concretizable if 
the following  formula is $\true$:
\begin{align}
\label{tempconcassume}
\exists c,\ldots,c_q \; \forall \varname_1,\ldots,\varname_q: \;\; 
(c_0 + \Sigma_{p=1}^q c_p*\varname_p \leq 0) \,  =  \, \conc(g).
\end{align}
If the above formula is $\true$, we extract values from the witness to the satisfiability 
of the inner $\forall$-formula, and say, 
\begin{align*}
c_0 + \Sigma_{p=1}^q c_p*\varname_p \leq 0 \; \in \; \Conc_{\tel}(\mtt{assume}\, (g)). 
\end{align*}
\item Similarly, the statement $\mtt{call} \; F_j(e_1,\ldots,e_\formalind)$ is 
concretizable if the following formula is $\true$: 
\begin{align*}
\exists c_{1,0},\ldots,c_{k,q} \; \forall \varname_1,\ldots,\varname_q: \;\;
\bigwedge_{i=1}^k \; ((c_{i,0} + \Sigma_{p=1}^q c_{i,p}*\varname_p \leq 0) \; = \; \conc(e_i)).
\end{align*}
If the above formula is $\true$, we can extract values from the witness to the satisfiability
of the inner $\forall$-formula to generate a concrete $\mtt{call}$ statement in 
$\Conc_{\tel}(\mtt{call} \; \cP_j(e_1,\ldots,e_\formalind))$.
\item The statement $\var_1,\ldots,\var_{\predind} \;
\bm{\assign} \; e_1,\ldots,e_\predind$  is concretizable if the formula in (\ref{tempconc}) is $\true$.
For convenience, let $h_j = c_{j,0} + \Sigma_{p=1}^q c_{j,p}*\varname_p$, for $j \in [1,q]$.
\begin{align}
\label{tempconc}
\exists c_{1,0},\ldots,c_{\predind,q} \; \forall \varname_1,\ldots, \varname_q: \;\;
\bigwedge_{i=1}^\predind \; \conc(\var_i)[\varname_1/h_1,\ldots,\varname_{q}/h_{q}] \; = \; \conc(e_i).
\end{align}
If the above formula is $\true$, we can extract values from the witness to the satisfiability
of the inner $\forall$-formula to generate a concrete assignment statement in $\Conc_{\tel}(\var_1,\ldots,\var_{\predind} \;
\bm{\assign} \; e_1,\ldots,e_\predind)$.
\end{enumerate}

\noindent{\em Example}: For our example in \figref{runexPBTG}, the modified
guards, $b0 \vee b1 \vee \neg b2$ and $b0 \vee b1 \vee b2$, in $stmt(\ell_1)$ and $stmt(\ell_2)$ of $\wcB$, respectively are concretized into
$\true$ and $x \leq 1$, respectively using $\conc$.\\

\noindent{\bf Concretization of inductive assertions}.
The concretization of each inductive assertion $\I_\ell \in \I_\Lambda$ is
simply $\conc(\I_\ell)$.

\section{Experiments with a Prototype Tool}\label{sec:expt}
\noindent We have built a prototype tool for repairing Boolean
programs.  The tool accepts Boolean programs generated by the
predicate abstraction tool SATABS (version 3.2) \cite{CKSY05} from
sequential C programs.  In our experience, we found that for C
programs with multiple procedures, SATABS generates (single procedure)
Boolean programs with all procedure calls inlined within the calling
procedure. Hence, we only perform intraprocedural analysis in this
version of our tool. The set of update schemas handled currently is
$\{id, \mtt{assign} \to \mtt{assign}, \mtt{assume} \to
\mtt{assume}\}$; we do not permit statement deletions. We set the
costs $c(\mtt{assign} \to \mtt{assign}, \ell)$ and $c(\mtt{assume} \to
\mtt{assume}, \ell)$  to some large number for every location $\ell$ where we wish
to disallow statement modifications, and to $1$ for all other
locations.  We initialize the tool with a repair budget of $1$.
We also provide the tool with a cut-set of locations
for its Boolean program input.

\begin{figure}[t]
\centering
%\subfloat[$\mtt{handmade1}$]{
\fbox{
\begin{minipage}{0.5\linewidth}
\footnotesize{
{\tt
$\mtt{handmade1}:$ \\
\\
int main() \{\\
\hspace*{3mm} $\mtt{int} \; x$;\\
\hspace*{3mm} $\ell_1:$ \hspace*{0mm} $\mtt{while}$ $(x < 0)$\\
\hspace*{3mm} $\ell_2:$ \hspace*{3mm} $x \; \assign \; x + 1$;\\
\hspace*{3mm} $\ell_3:$ $\mtt{assert}$ $(x>0)$;\\
$\}$\\
}
\noindent\rule{6cm}{0.4pt}
\\
Boolean program vars/predicates:
\begin{enumerate}
\item $\conc(b0) =  x \leq 0$
\end{enumerate}
\noindent\rule{6cm}{0.4pt}
\\
Boolean program repair:
\begin{enumerate}
\item Change guard for $stmt(\ell_1)$ from
$*$ to $b0$
\end{enumerate}
\noindent\rule{6cm}{0.4pt}
\\
Concrete program repair:
\begin{enumerate}
\item Change guard for $stmt(\ell_1)$ to
$x \leq 0$
\end{enumerate}
}
\end{minipage}
}
%}
%%
%\subfloat[$\mtt{handmade2}$]{
%\fbox{
%\begin{minipage}{0.6\linewidth}
%\footnotesize{
%{\tt
%%$\mtt{handmade2}:$ \\
%%\\
%int main() \{\\
%\hspace*{3mm} $\mtt{int} \; x$;\\
%\hspace*{3mm} $\ell_1:$ $\mtt{if}$ $(x\leq0)$\\
%\hspace*{3mm} $\ell_2:$ \hspace*{3mm} $\mtt{while}$ $(x < 0)$\{\\
%\hspace*{3mm} $\ell_3:$ \hspace*{6mm} $x \; \assign \; x + 2$;\\
%\hspace*{3mm} $\ell_4:$ \hspace*{6mm} $\mtt{skip}$;\\
%\hspace*{15mm} $\}$\\
%\hspace*{10mm} $\mtt{else}$\\
%\hspace*{3mm} $\ell_5:$ \hspace*{3mm} $\mtt{if}$ $(x==1)$\\
%\hspace*{3mm} $\ell_6:$ \hspace*{6mm} $x \; \assign \; x - 1$;\\
%\hspace*{3mm} $\ell_7:$ $\mtt{assert}$ $(x>1)$;\\
%$\}$\\
%}
%\noindent\rule{8cm}{0.4pt}
%\\
%Boolean program vars/predicates:\\
%$\conc(b0) \; =  \; x \leq 1$\\
%$\conc(b1) \; = \; x == 1$ \\
%$\conc(b2) \; = \;  x \leq 0$\\
%\\
%Boolean program repair:\\
%Change guard for $stmt(\ell_1)$
%from $b2$ to $b0 \vee b1 \vee \neg b2$\\
%Change guard for $stmt(\ell_2)$
%from $*$ to $b0 \vee b1 \vee b2$\\
%\\
%Concrete program repair:\\
%Change guard for $stmt(\ell_1)$ to $\true$\\
%Change guard for $stmt(\ell_2)$ to $x \leq 1$\\
%}
%\end{minipage}
%}
%}
\caption{Repairing program $\mtt{handmade1}$}
\label{fig:seq.pred.ex1}
\end{figure}

\begin{figure}[!htb]
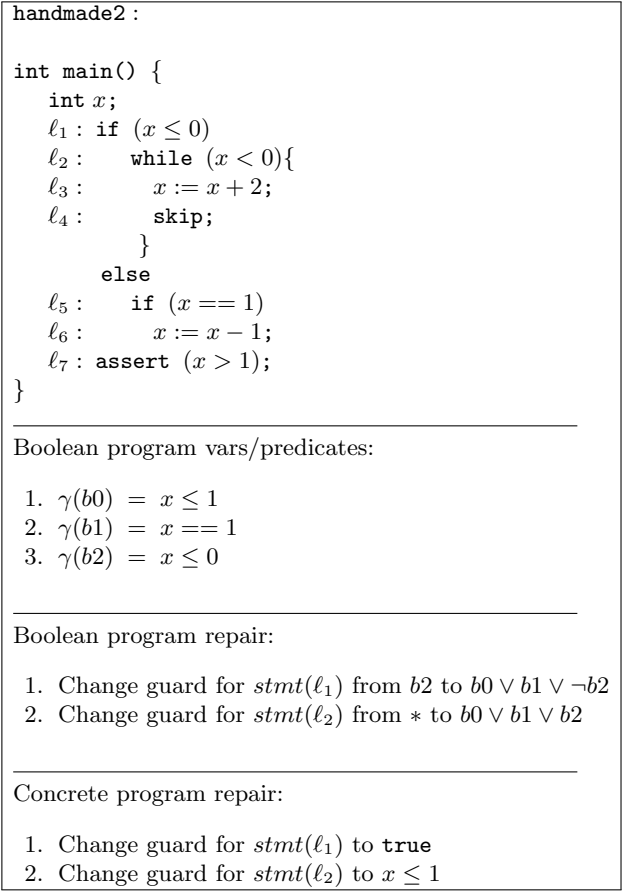

\centering
\fbox{
\begin{minipage}{0.65\linewidth}
\footnotesize{
{\tt
$\mtt{handmade2}:$ \\
\\
int main() \{\\
\hspace*{3mm} $\mtt{int} \; x$;\\
\hspace*{3mm} $\ell_1:$ $\mtt{if}$ $(x\leq0)$\\
\hspace*{3mm} $\ell_2:$ \hspace*{3mm} $\mtt{while}$ $(x < 0)$\{\\
\hspace*{3mm} $\ell_3:$ \hspace*{6mm} $x \; \assign \; x + 2$;\\
\hspace*{3mm} $\ell_4:$ \hspace*{6mm} $\mtt{skip}$;\\
\hspace*{15mm} $\}$\\
\hspace*{10mm} $\mtt{else}$\\
\hspace*{3mm} $\ell_5:$ \hspace*{3mm} $\mtt{if}$ $(x==1)$\\
\hspace*{3mm} $\ell_6:$ \hspace*{6mm} $x \; \assign \; x - 1$;\\
\hspace*{3mm} $\ell_7:$ $\mtt{assert}$ $(x>1)$;\\
$\}$\\
}
\noindent\rule{7.5cm}{0.4pt}
\\
Boolean program vars/predicates:
\begin{enumerate}
\item $\conc(b0) \; =  \; x \leq 1$
\item $\conc(b1) \; = \; x == 1$ 
\item $\conc(b2) \; = \;  x \leq 0$
\end{enumerate}
\noindent\rule{7.5cm}{0.4pt}
\\
Boolean program repair:
\begin{enumerate}
\item Change guard for $stmt(\ell_1)$
from $b2$ to $b0 \vee b1 \vee \neg b2$
\item Change guard for $stmt(\ell_2)$
from $*$ to $b0 \vee b1 \vee b2$
\end{enumerate}
\noindent\rule{7.5cm}{0.4pt}
\\
Concrete program repair:
\begin{enumerate}
\item Change guard for $stmt(\ell_1)$ to $\true$
\item Change guard for $stmt(\ell_2)$ to $x \leq 1$
\end{enumerate}
}
\end{minipage}
}
\caption{Repairing program $\mtt{handmade2}$}
\label{fig:seq.pred.ex2}
\end{figure}

Given the above, the tool automatically generates an SMT query
corresponding to the inner $\forall$-formula in (\ref{PartCorrect}).
When generating this repairability query, for update schemas involving expression
modifications, we stipulate every deterministic Boolean expression
$g$ be modified into an {\em unknown} deterministic Boolean expression
$f$ (as described in \figref{updef}), and every
nondeterministic Boolean expression be modified into an unknown
nondeterministic expression of the form $\mtt{choose}(f_1,f_2)$.
The SMT query is then fed to the SMT-solver Z3 (version 4.3.1)
\cite{MouBjo08}. The solver either declares the formula to be
satisfiable, and provides models for all the unknowns, or declares the
formula to be unsatisfiable. In the latter case, we can choose to
increase the repair budget by $1$, and repeat the process.

\begin{figure}[!htb]
\centering
\fbox{
\begin{minipage}{0.65\linewidth}
\footnotesize{
{\tt
$\mtt{necex6}:$\\
\\
$\mtt{int} \; x,y$;\\
\\
int foo(int $*ptr$) $\{$\\
\hspace*{3mm} $\ell_4:$ if ($ptr == \&x$) \\
\hspace*{3mm} $\ell_5:$ \hspace*{3mm} $*ptr \; \assign \; 0$;\\
\hspace*{3mm} $\ell_6:$ if ($ptr == \&y$)\\
\hspace*{3mm} $\ell_7:$ \hspace*{3mm} $*ptr \; \assign \; 1$;\\
\hspace*{3mm} return $1$;\\
$\}$\\
\\
int main() $\{$\\
\hspace*{3mm} $\ell_1:$ foo ($\&x$);\\
\hspace*{3mm} $\ell_2:$ foo ($\&y$);\\
\hspace*{3mm} $\ell_3:$ assert ($x>y$);\\
$\}$\\
}
\noindent\rule{7.5cm}{0.4pt}
\\
Boolean program vars/predicates:
\begin{enumerate}
\item $\conc(b0) \; =  \; y < x$
\item $\conc(b1) \; = \; \&y == ptr$ 
\item $\conc(b2) \; = \; \&x == ptr$
\end{enumerate}
\noindent\rule{7.5cm}{0.4pt}
\\
Boolean program repair:
\begin{enumerate}
\item Change $stmt(\ell_7)$ from
$b0 \; \assign \; *$ to
$b0 \; \assign \; b0 \vee b1 \vee b2$
\end{enumerate}
\noindent\rule{7.5cm}{0.4pt}
\\
Concrete program repair:
\begin{enumerate}
\item Change $stmt(\ell_7)$ to
$*ptr \; \assign \; -1$
\end{enumerate}
}
\end{minipage}
}
\caption{Repairing program $\mtt{necex6}$}
\label{fig:seq.pred.ex3}
\end{figure}

Once the solver provides models for all the unknowns, we can extract a
repaired Boolean program. Currently, the next step --- concretization
--- is only partly automated. For assignment statements, we manually
formulate SMT queries corresponding to the inner $\forall$-formula in
(\ref{conc}), and feed these queries to Z3. If the relevant queries
are found to be satisfiable, we can obtain a repaired C program. If
the queries are unsatisfiable, we attempt template-based
concretization using linear-arithmetic templates. We manually
formulate SMT queries corresponding to the inner $\forall$-formulas in
(\ref{tempconcassume}) and (\ref{tempconc}), and call Z3. In some of
our experiments, we allowed ourselves a degree of flexibility in
guiding the solver to choose the right template parameters.

In \figref{seq.pred.ex1}, \figref{seq.pred.ex2}, \figref{seq.pred.ex3} and
\figref{seq.pred.ex4}, we present some of the details of repairing four C
programs.  The first two programs are handmade, with the second one being the
same as the one shown in \figref{runexPBTG}.  The next two programs are
mutations of two programs drawn from the NEC Laboratories Static Analysis
Benchmarks \cite{NECBench}.

\begin{figure}[h]
\centering
\fbox{
\begin{minipage}{0.7\linewidth}
\footnotesize{
{\tt
$\mtt{necex14}:$\\
\\
int main() $\{$\\
\hspace*{3mm} int $x,y$;\\
\hspace*{3mm} int $a[10]$;\\
\hspace*{3mm} $\ell_1:$ $x \; \assign \; 1U$;\\
\hspace*{3mm} $\ell_2:$ while ($x \leq 10U$) $\{$\\
\hspace*{3mm} $\ell_3:$ \hspace*{3mm} $y \; \assign \; 11-x$;\\
\hspace*{3mm} $\ell_4:$ \hspace*{3mm} assert ($y \geq 0 \; \wedge \;
y<10$);\\
\hspace*{3mm} $\ell_5:$ \hspace*{3mm} $a[y] \; \assign \; -1$;\\
\hspace*{3mm} $\ell_6:$ \hspace*{3mm} $x \; \assign \; x + 1$;\\
\hspace*{9mm} $\}$\\
$\}$\\
}
\noindent\rule{8cm}{0.4pt}
\\
Boolean program vars/predicates:
\begin{enumerate}
\item $\conc(b0) \; =  \; y < 0 $
\item $\conc(b1) \; =  \; y < 10 $
\end{enumerate}
\noindent\rule{8cm}{0.4pt}
\\
Boolean program repair:
\begin{enumerate}
\item Change $stmt(\ell_3)$ from $b0,b1 \; \assign \; *,*$ to 
$b0,b1 \; \assign \; \mtt{F},\mtt{T}$
\end{enumerate}
\noindent\rule{8cm}{0.4pt}
\\
Concrete program repair:
\begin{enumerate}
\item Change $stmt(\ell_3)$ to $y \; \assign \; 10 - x$
\end{enumerate}
}
\end{minipage}
}
\caption{Repairing program $\mtt{necex14}$}
\label{fig:seq.pred.ex4}
\end{figure}

We emphasize that the repairs for the respective Boolean programs (not
shown here due to lack of space) are obtained automatically.  The
concretization of the repaired Boolean program in
\figref{seq.pred.ex1} was trivial -- it only involved concretizing the
guard $b0$ corresponding to the statement at location $\ell_1$.
Concretization of the repaired Boolean program in
\figref{seq.pred.ex2} involved concretizing two different guards, $b0
\vee b1 \vee \neg b2$ and $b0 \vee b1 \vee b2$, corresponding to the
statements at locations $\ell_1$ and $\ell_2$, respectively.  We
manually simplified the concretized guards to obtain the concrete
guards $\true$ and $x \leq 1$, respectively.  Concretization of the
repaired Boolean program in \figref{seq.pred.ex3} involved
concretizing the assignment statement at location $\ell_7$.  We
manually formulated an SMT query corresponding to the formula in
(\ref{conc}), after simplifying $\conc(b_0 \vee b_1 \vee b_2)$ to $y <
x$ and restricting the LHS of $stmt(\ell_7)$ in the concrete program
to remain unchanged.  The query was found to be satisfiable, and
yielded $-1$ as the RHS of the assignment statement in the concrete
program.  We repeated the above exercise to concretize the assignment
statement at location $\ell_3$ in \figref{seq.pred.ex4}, and obtained
$y \; \assign \; 0$ as the repair for the concrete program.
Unsatisfied by this repair, we formulated another SMT query
corresponding to the formula in (\ref{tempconc}), restricting the RHS
of $stmt(\ell_3)$ to the template $-x + c$, where $c$ is unknown. The
query was found to be satisfiable, and yielded $c = 10$.

In \tabref{results}, we present the results of repairing the above four programs 
and some benchmark programs from the 2014 Competition on
Software Verification \cite{SVComp14}.  The complexity of the programs from
\cite{SVComp14} stems from nondeterministic assignments and function
invocations within loops.  All experiments were run on the same machine, an
Intel Dual Core 2.13GHz Unix desktop with 4 GB of RAM.

We enumerate the time taken for each individual step involved in generating a
repaired Boolean program.  The columns labeled LoC($\cP$) and LoC($\cB$)
enumerate the number of lines of code in the original C program and the Boolean
program generated by SATABS, respectively. The column labeled $\Var(\cB)$
enumerates the number of variables in each Boolean program. The column
$\cB$-time enumerates the time taken by SATABS to generate each Boolean
program, the column Que-time enumerates the time taken by our tool to generate
each repairability query and the column Sol-time enumerates the time taken by
Z3 to solve the query. The columns $\#$ $\mtt{Asg}$ and $\#$ $\mtt{Asm}$ count
the number of $\mtt{assign} \to \mtt{assign}$ and $\mtt{assume} \to
\mtt{assume}$ update schemas applied, respectively, to obtain the final correct
program.

Notice that our implementation either produces a repaired program very quickly,
or fails to do so in reasonable time whenever there is a significant increase
in the number of Boolean variables, as was the case for example, in {\tt
veris.c\_NetBSD-libc\_\_loop\_true}. This is because the SMT solver might need
to search over simultaneous non-deterministic assignments to all the Boolean
variables for every assignment statement in $\cB$ in order to solve the
repairability query. For the last two programs, SATABS was the main bottleneck,
with SATABS failing to generate a Boolean program with a non-spurious
counterexample after 10 minutes. In particular, we experienced issues while using
SATABS on programs that relied heavily on character manipulation.

We emphasize that when successful, our tool can repair a diverse set of errors
in programs containing loops, multiple procedures and pointer and array
variables.  In our benchmarks, we were able to repair operators (e.g., an
incorrect conditional statement $x < 0$ was repaired to $x > 0$) and array
indices (e.g., an incorrect assignment $x \assign a[0]$ was repaired to $x
\assign a[j])$, and modify constants into program variables (e.g. an incorrect
assignment $x \assign 0$ was repaired to $x \assign d$, where $d$ was a program
variable). Also, note that for many benchmarks, the repaired programs required
multiple statement modifications.

\begin{table}[t]

\caption{Experimental results}

\scalebox{0.8}{

\centering
\begin{tabular}{|l||c|c|c|c|c|c|l|c|c}
\hline
Name & LoC($\cP$) & LoC($\cB$) & $\Var(\cB)$ & $\cB$-time &
Que-time & Sol-time & \# Asg & \# Asm \\
\hline
$\mtt{handmade1}$ & 6 & 58 & 1 & 0.180s & 0.009s & 0.012s & 0 & 1
\\
$\mtt{handmade2}$ & 16 & 53 & 3 & 0.304s & 0.040s & 0.076s & 0 & 2
\\
\hline
$\mtt{necex6}$ & 24 &  66  & 3 & 0.288s & 0.004s &  0.148s & 1 & 0
\\
$\mtt{necex14}$ & 13 &  60 & 2 & 0.212s & 0.004s & 0.032s & 1 & 0
\\
%$\mtt{tcas1}$ & 283 & 856 & 49 & 17m58.675s & 0.308s & $>$ 5h & - & -
%\\
\hline
$\mtt{while\_infinite\_loop\_1\_true}$ & 5 & 33 & 1 & 0.196s & 0.002s & 0.008s & 0 & 1
\\
$\mtt{array\_true}$ & 23 & 57 & 4 & 0.384s & 0.004s & 0.116s & 1 & 1
\\
$\mtt{n.c11\_true}$ & 27 & 50 & 2 & 0.204s & 0.002s & 0.024s & 1 & 0
\\
$\mtt{terminator\_03\_true}$ & 22 & 38 & 2 & 0.224s & 0.004s & 0.036s & 1 & 1
\\
$\mtt{trex03\_true}$ & 23 & 58 & 3 & 0.224s & 0.036s & 0.540s & 1 & 1
\\
$\mtt{trex04\_true}$ & 29 & 36 & 1 & 0.200s & 0.004s & 0.004s & 2 & 0
\\
$\mtt{veris.c\_NetBSD-libc\_\_loop\_true}$ & 30 & 144 & 23 & 3.856s & - & - & - & -
\\
$\mtt{vogal\_true}$ & 41 & - & - & $>10m$ & - & - & - & -
\\
$\mtt{count\_up\_down\_true}$ & 18 & - & - & $>10m$ & - & - & - & -
\\
\hline
\end{tabular}
}
\vspace{10pt}
\label{tab:results}
\end{table}

\section{Discussion}\label{sec:exten}
While the algorithm presented in this paper separates the computation of a
repaired Boolean program $\wcB$ from its concretization to obtain $\wcP$, this
separation is not necessary.  In fact, the separation may be sub-optimal - it
may not be possible to concretize all modified statements of a computed $\wcB$,
while there may indeed exist some other concretizable $\wcB$.  The solution is
to directly search for $\wcB$ such that all modified statements of $\wcB$ are
concretizable. This can be done by combining the constraints presented in
\secref{algo2} with the one in (\ref{PartCorrect}). In particular,
the set $Unknown$ in (\ref{PartCorrect}) can be modified to include unknown
expressions/template parameters needed in the formulas in
\secref{algo2}, and $\CRC(\pi)$ can be modified to include the
inner quantifier-free constraints in the formulas in
\secref{algo2}.

As noted in \secref{intro}, we can target total correctness of the repaired
programs by associating ranking functions along with inductive assertions
with each cut-point in $\Lambda$, and including termination conditions as part
of the constraints.

Finally, we wish to explore ways to ensure that the repaired program does not
unnecessarily restrict correct behaviors of the original program. We
conjecture that this can be done by computing the weakest possible set of
inductive assertions and a least restrictive $\wcB$.

%\section{Related Work}\label{sec:relwork}
%\input{relatedwork.tex}

%\section{Discussion}\label{sec:conc}
%\input{conclusion.tex}

\bibliographystyle{splncs03}
\bibliography{myrefs}

%\newpage

\normalsize
%\appendices

%\section{Operational Semantics}\label{app:app1}
%\input{app1.tex}

\end{document}